\documentclass[10pt,twocolumn,twoside]{IEEEtran}

\usepackage{generic}
\usepackage{amsmath,amssymb,amsfonts}
\DeclareMathOperator{\Tr}{Tr}
\usepackage{graphicx}
\usepackage{textcomp}

\def\BibTeX{{\rm B\kern-.05em{\sc i\kern-.025em b}\kern-.08em
    T\kern-.1667em\lower.7ex\hbox{E}\kern-.125emX}}
    
\usepackage{amsthm}
\usepackage{graphics} 
\usepackage{epsfig} 
\usepackage{times} 

\usepackage[sorting=none,maxbibnames=99,giveninits=true]{biblatex}
\usepackage{bm}
\usepackage{xcolor}
\usepackage{booktabs}
\usepackage{siunitx}
\usepackage{hyperref}
\usepackage{cleveref}

\usepackage{caption}
\usepackage{lipsum}
\usepackage{multirow}
\usepackage{diagbox} 
\usepackage{subcaption}

\newcommand{\R}{\mathbb{R}}
\newcommand{\E}{\mathbb{E}}
\newcommand{\Pro}{\mathbb{P}}
\newcommand{\N}{\mathcal{N}}

\newcommand{\erf}{\textup{erf}}

\newcommand{\cass}{\mathcal{A}^{i,j}_{\varepsilon}}
\newcommand{\casm}{\mathcal{A}^{\mathcal{I}_m,j}_{\varepsilon}}
\newcommand{\casa}{\mathfrak{A}_\varepsilon^{\mathcal{I}_m,j}}

\newcommand{\casv}{\mathfrak{R}_\varepsilon^{\mathcal{I}_m,j}}

\newcommand{\rpro}{\bm{\mathcal{A}}^{\mathcal{I}_m}_{\varepsilon}}
\newcommand{\sigi}{\sigma_i}
\newcommand{\sigj}{\sigma_j}

\newcommand{\tm}{\tilde{\mu}}
\newcommand{\ts}{\tilde{\sigma}}

\usepackage[T1]{fontenc}
\newcommand{\textincon}[1]{%
{\fontfamily{zi4}\selectfont #1}}
\newcommand{\VAR}{{\textrm{\large \textincon{V$@$R}}}}
\newcommand{\AVAR}{{\textrm{\large \textincon{AV$@$R}}}}

\newtheorem{definition}{\bf Definition}
\newtheorem{assumption}{\bf Assumption}
\newtheorem{theorem}{\bf Theorem}
\newtheorem{lemma}{\bf Lemma}
\newtheorem{corollary}{\bf Corollary}

\title{\LARGE \bf Incremental Risk Assessment for Cascading Failures in Large-Scale Multi-Agent Systems}

\author{Guangyi Liu$^{1}$, Vivek Pandey$^{2}$, Christoforos Somarakis$^{3}$, and Nader Motee$^{2}$%

\thanks{$^{1}$Guangyi Liu is with Amazon Robotics, North Reading, MA, USA. {\tt\small gyliu@amazon.com}. This paper is independent of his position at Amazon and does not relate to his employment there.}%

\thanks{$^{2}$Vivek Pandey and Nader Motee are with the Department of Mechanical Engineering and Mechanics, Lehigh University, Bethlehem, PA, USA. {\tt\small \{vkp219,motee\}@lehigh.edu}.}%

\thanks{$^{3}$Christoforos Somarakis is a Senior Scientist with the Applied Mathematics Group, Merck \& Co., USA. {\tt\small christoforos.somarakis@merck.com}.}%
}

\begin{document}

\maketitle

\begin{abstract}
We develop a  framework for studying and quantifying the risk of cascading failures in time-delay consensus networks, motivated by a team of agents attempting temporal rendezvous under stochastic disturbances and communication delays. To assess how failures at one or multiple agents amplify the risk of deviation across the network, we employ the Average Value-at-Risk as a systemic measure of cascading uncertainty. Closed-form expressions reveal explicit dependencies of the risk of cascading failure on the Laplacian spectrum, communication delay, and noise statistics. We further establish fundamental lower bounds that characterize the best-achievable network performance under time-delay constraints. These bounds serve as feasibility certificates for assessing whether a desired safety or performance goal can be achieved without exhaustive search across all possible topologies. In addition, we develop an efficient single-step update law that enables scalable propagation of conditional risk as new failures are detected. Analytical and numerical studies demonstrate significant computational savings and confirm the tightness of the theoretical limits across diverse network configurations.
\end{abstract}

\section{Introduction}

Consensus networks are fundamental to a wide range of applications, from opinion formation in social systems to coordination in engineered multi-agent systems. For instance, in human societies, individuals often form beliefs and make decisions based on perceived social consensus~\cite{krueger1998perception}. In engineered settings, such as robotic teams, agents coordinate their behavior by following common protocols that promote group agreement~\cite{fagiolini2008consensus}. Despite their broad utility, consensus processes are inherently vulnerable to imperfections such as communication delays, limited sensing, and external disturbances. These factors can cause individual agents to diverge from the group consensus and degrade overall system performance.

Much of the existing literature has focused on the probability of consensus failure caused by such disturbances~\cite{Somarakis2019g}. However, an equally important and less explored question is how such deviations propagate through a network. Traditionally, a ``cascading failure'' describes a domino effect analyzed strictly \emph{after} a hard failure has occurred at a specific node. In this paper, we generalize the concept of a cascade to the continuous domain. We define a \emph{fluctuation cascade} as the conditional amplification and propagation of large deviations across a network under partial or range-bounded information. Rather than only asking what happens when a hard failure has already occurred, we ask: \emph{How does an agent entering an unsafe alarm zone amplify the conditional risk of failure for the rest of the system?} By capturing this continuous risk propagation, our framework models the precursors to system-wide failures, encompassing traditional post-failure cascades as a special, limiting case.

Uncertainty is intrinsic to physical systems, from quantum particles to large-scale engineered networks. As such, failures in consensus systems are not merely possible but inevitable over time~\cite{7438924,zhang2018cascading,zhang2019robustness,9683468,liu2022emergence, liu2025risk}. This motivates our interest in analyzing the \emph{resilience} of consensus networks in the presence of existing failures. For example, how does a single malfunctioning robot affect the performance of a coordinated swarm? Or, in a social context, how do committed opinions influence the ability of the network to reach consensus~\cite{xie2011social}? From a systems perspective, understanding these effects is essential for designing robust control and decision-making architectures that can isolate or contain the spread of failures.

In this paper, we develop a theoretical framework grounded in systemic \emph{risk analysis} \cite{rockafellar2000optimization,rockafellar2002conditional} to evaluate the likelihood and severity of cascading failures in time-delay consensus networks. Our goal is to quantify how failures occurring at one or multiple agents propagate through the network and elevate the risk of deviation in other agents. This risk-based formulation provides actionable insights into the design of resilient networked systems, enabling systematic evaluation of their ability to withstand and localize the effects of partial failures.

As a motivating case study, we consider a team of autonomous agents attempting to reach consensus on a rendezvous time. Agents exchange information over a time-invariant communication graph subject to uniform time delays and independent stochastic disturbances. These delays and noise sources model real-world limitations such as sensor latency and environmental uncertainty. Our analysis focuses on the event where agents fails to reach agreement, and we study how this deviation alters the risk profile for other agents in the network.

\textit{Our Contributions:} Building upon our previous work on first-order consensus networks~\cite{Somarakis2019g,Somarakis2016g,Somarakis2017a}, we extend the notion of individual deviation risk to \emph{cascading} risk. Specifically:
\begin{itemize}
    \item We introduce a formal framework using Average Value-at-Risk ($\AVAR$) to quantify the risk of cascading failures in stochastic consensus networks with time-delay.
    \item We derive closed-form expressions for the conditional risk of large deviations, given that one or more agents have already failed to reach consensus. These expressions explicitly capture the effects of network topology, time-delay, and noise.
    \item We analyze how uncertainty in individual nodes and their interactions contributes to overall risk via marginal variance and pairwise correlation structures.
    \item We validate our theoretical findings through simulation studies on canonical graph topologies (e.g., path, star, and complete graphs), revealing how structural features influence the system's vulnerability.
\end{itemize}

\noindent \textit{Relation to Prior Work:} 
Compared with our earlier conference works~\cite{liu2022risk,liu2023cascading}, which focused on evaluating conditional cascading risk for a given network, this paper makes two substantive generalizations. First, we establish time-delay–induced fundamental limits and derive a universal lower bound on the best-achievable cascading risk (\Cref{thm:risk_bound}), providing a feasibility certificate independent of specific topologies. Second, we develop a single-step incremental update rule (\Cref{thm:conditional_prob_update}) that efficiently propagates conditional mean and variance as new failures are observed, enabling scalable online risk re-evaluation without recomputing high-dimensional inverses from scratch.

All the proofs of theoretical results are provided in the appendix.

\section{Mathematical Notation}

We denote the non-negative orthant of $\R^n$ by $\R_+^n$, the standard basis by $\{\bm{e}_1, \dots, \bm{e}_n\}$, and the all-ones vector by $\bm{1}_n = [1,\dots,1]^\top$. The $n \times n$ identity matrix is $I_n$.

Let $\mathcal{G}$ be a simple, connected, undirected, and weighted graph with Laplacian $L = [l_{ij}] \in \R^{n \times n}$ defined by
\[
    l_{ij} := \begin{cases}
        -k_{ij}, & i \neq j \\
        \sum_{j \neq i} k_{ij}, & i = j
    \end{cases},
\]
where $k_{ij} \geq 0$ is the edge weight. The matrix $L$ is symmetric and positive semi-definite~\cite{van2010graph}, with eigenvalues $0 = \lambda_1 < \lambda_2 \le \dots \le \lambda_n$. Let $Q = [\bm{q}_1|\dots|\bm{q}_n]$ be the orthonormal eigenvector matrix satisfying $Q^\top Q = I_n$ and $\bm{q}_1 = \bm{1}_n / \sqrt{n}$. Then $L = Q \Lambda Q^\top$ with $\Lambda = \text{diag}(0, \lambda_2, \dots, \lambda_n)$.

Let $\mathcal{L}^2(\R^q)$ be the space of $\R^q$-valued random vectors with finite second moments on $(\Omega, \mathcal{F}, \mathbb{P})$. A Gaussian vector $\bm{y} \sim \mathcal{N}(\bm{\mu}, \Sigma)$ has mean $\bm{\mu} \in \R^q$ and covariance $\Sigma \in \R^{q \times q}$. The error function $\mathrm{erf}:\R \to (-1,1)$ is
$\mathrm{erf}(x) = \frac{2}{\sqrt{\pi}} \int_0^x e^{-t^2} \mathrm{d}t,$
with inverse $\mathrm{erf}^{-1}(\cdot)$.

\section{Problem Statement}     \label{sec:problem_statement}

We consider a class of time-delay linear consensus networks that arise in engineering applications such as clock synchronization in sensor networks, time or spatial rendezvous, and heading alignment in swarms; see~\cite{ren2007information,olfati2007consensus,saldana2018modquad} for details. As an application, we consider the rendezvous problem in time where the group objective is to meet simultaneously at a prespecified location known to all agents.\footnote{Rendezvous in space is very similar to rendezvous in time by switching the role of time and location.} Agents do not have prior knowledge of the precise meeting time as it may have to be adjusted in response to unexpected emergencies or exogenous uncertainties \cite{ren2007information}. Thus, all agents should agree on a rendezvous time by achieving the consensus, which can be accomplished by each agent $i = 1, \dots, n$ creating a state variable, say $x^{(i)} \in \R$, representing its belief of the rendezvous time. Each agent's initial belief is set to its preferred time that it can rendezvous with others. Then, the rendezvous dynamics for each agent evolves in time according to the following stochastic differential equations:
\begin{equation} \label{eq:dyn-one}
    \text{d}  x^{(i)}_t = u^{(i)}_t \, \text{d}  t + b \, \text{d}  w^{(i)}_t,
\end{equation}
for all $i = 1, \dots, n$. Each agent's control input is $u^{(i)}_t \in \R$. The source of uncertainty is diffused in the network as additive stochastic noise, and its magnitude is uniformly scaled by the diffusion coefficient $b \in \R_{+}$. The impact of uncertain environments on dynamics of agents are modeled by independent Brownian motions $w^{(1)}, \dots, w^{(n)}$. 
In many real-world systems, such as multi-robot teams using motion capture for coordination, agents receive updates via wireless broadcast from a central processor, leading to near-identical communication latency. Motivated by this, we assume that all agents experience an identical communication time-delay $\tau \in \mathbb{R}_{+}$~\cite{ren2007information}. The control inputs are determined via a negotiation process by forming a linear consensus network over a communication graph using the following feedback law:
\begin{equation} \label{eq:feedback}
    u^{(i)}_t = \sum_{j = 1}^{n} k_{ij} \left( x^{(j)}_{t-\tau} - x^{(i)}_{t-\tau} \right),
\end{equation}
where $k_{ij} \in \R_{+}$ are nonnegative feedback gains. Let us denote the state vector by $\bm{x}_t = [ x^{(1)}_{t}, \dots, x^{(n)}_{t} ]^\top$ and the vector of exogenous disturbance by $\bm{w}_t = [w^{(1)}_{t}, \dots, w^{(n)}_{t} ]^\top$. The dynamics of the resulting closed-loop network can be cast as a linear consensus network that is governed by the following stochastic differential equation:
\begin{equation} \label{eq:dyn}
    \text{d}  \bm{x}_t = -L \, \bm{x}_{t-\tau}\, \text{d}  t + B \, \text{d}  \bm{w}_t,
\end{equation}
for all $t \geq 0$, where the initial function $\bm{x}_t=\phi(t)$ is deterministically given for $t \in [-\tau, 0]$ and $B = b I_n$. The underlying coupling structure of the consensus network \eqref{eq:dyn} is a connected graph $\mathcal{G}$ with Laplacian matrix $L$. It is considered that the communication graph $\mathcal{G}$ is time-invariant such that the network of agents aim to reach the consensus on a rendezvous time before they perform motion planning to get to the meeting location. Upon reaching consensus, a properly designed internal feedback control mechanism steers each agent toward the rendezvous location.

\begin{assumption}  \label{asp:stable}
    The time-delay satisfies $\tau < \frac{\pi}{2 \lambda_n}$.
\end{assumption}

When there is no noise, i.e., $b = 0$, it is known \cite{olfati2004consensus} that under the Assumption \ref{asp:stable} and graph being connected, states of all agents converge to the average of all initial states $\frac{1}{n} \bm{1}_n^\top \bm{x}_0$; whereas in the presence of input noise, state variables fluctuate around the network average $\frac{1}{n} \bm{1}_n^\top \bm{x}_t$. In order to quantify the quality of rendezvous and its fragility features, we consider the vector of observables 
\begin{equation} \label{eq:observables}
     \bm{y}_t = M_n \, \bm{x}_t,
\end{equation}
in which $M_n = I_n - \frac{1}{n} \bm{1}_n \bm{1}_n^\top $ is the centering matrix and the observable $\bm{y}_t = [y^{(1)}_t, ..., y^{(n)}_t]^\top$ measures the agents' deviations from the current network average. The assumption of connected graph implies that one of the modes of network \eqref{eq:dyn} is marginally stable. The marginally stable mode, which corresponds to the zero eigenvalue of $L$, is unobservable from the output \eqref{eq:observables}, which keeps $\bm{y}_t$ bounded in the steady-state. When noise is absent, we have $\bm{y}_t \rightarrow 0$ as $t \rightarrow \infty$. Consequently, the exogenous noise excites the observable modes of the network and the output fluctuates around zero. This implies that agents will not agree upon an exact rendezvous time and a practical resolution is to allow a tolerance interval for agents to concur.

\begin{definition}
    For a given $c \in \R_+$, the network \eqref{eq:dyn} reaches the $c$-consensus if, in steady state where $\bm{y}_t \Rightarrow \bm{\bar y}$ as $t \to \infty$,
    \begin{equation} \label{eq:event}
        |\bm{\bar y}| \leq c \bm{1}_n
    \end{equation}
    holds with a high probability\footnote{The high probability means a probability larger than a predefined cut-off number close to one.}.
\end{definition}

The notion of $c$-consensus means that all agents have agreement on all points in $\{\bm{x} \in \R^n \,\big|\,  |M_n \bm{x}| \leq c \bm{1}_n\}$. Suppose that event \eqref{eq:event} holds, the network of agents will achieve a $c$-consensus of the rendezvous time in the following sense. In steady-state, the $i$'th agent is assured that by $x^{(i)}_t \pm \,c$ units of time, all other agents will arrive and meet each other in that time interval with high probability. At the same time, some undesirable situations may also occur that we refer to as failures.

\begin{definition}
    For a given $c \in \R_+$, an agent whose motion is governed by \eqref{eq:dyn} with steady-state observable $\bar{y}_i$ defined in \eqref{eq:observables} is said to be prone to failure if
    \begin{equation} \label{eq:failure}
        \Pro\big(|\bar{y}_i| > c\big) > 0.
    \end{equation}
\end{definition}

In the rendezvous problem, a failure event \eqref{eq:failure} with probability exceeding $\varepsilon > 0$ indicates that one or more agents deviate significantly from the consensus, potentially preventing the entire network from achieving $c$-consensus within the intended rendezvous interval and may trigger cascading failures among the remaining agents.

The {\it problem} is to quantify the risk of such cascading failures, i.e., large deviations conditioned on the failure of other agents, as a function of the graph Laplacian, time-delay, and noise statistics. To this end, we develop a systemic risk framework based on the steady-state behavior of the closed-loop stochastic system.

The remainder of the paper is organized as follows. Section~\ref{sec:prelims} reviews the steady-state behavior of time-delay consensus networks. Section~\ref{sec:risk} formulates cascading-failure risk via closed-form conditional tail metrics. Section~\ref{sec:special-graphs} specializes these results to canonical topologies and highlights topology-dependent behaviors. Section~\ref{sec:update} presents a scalable single-step update law for propagating risk of cascading failures as new failures are observed. Section~\ref{specialgraph} derives time-delay–induced fundamental limits and best-achievable lower bounds on risk. Section~\ref{sec:case-study} validates the theory through simulations on representative networks and discusses design implications.

\section{Preliminary Results}   \label{sec:prelims}

We begin by characterizing the steady-state statistics of the network observables and introducing risk metrics that quantify the severity of large deviations.

\subsection{Steady-State Statistics of Observables}
Under Assumption~\ref{asp:stable} and a connected communication graph, the steady-state observables \eqref{eq:observables} converge in distribution to a multivariate normal \cite{Somarakis2019g,liu2022risk}, $\bm{\bar{y}} \sim \mathcal{N}(0, \Sigma)$. The closed-form expression for the covariance matrix $\Sigma$ is provided below.
\begin{lemma}     \label{lem:y_steady}
    The steady-state covariance matrix of $\bm{\bar{y}}$, $\Sigma = [\sigma_{ij}]$, is given element-wise by
    \begin{equation} \label{eq:sigma_y}
    \begin{aligned}
        \sigma_{ij} = 
        \frac{1}{2} b^2 \sum_{k=2}^{n} \frac{\cos (\lambda_k \tau)}{\lambda_k (1 - \sin (\lambda_k \tau))} (\bm{m}_i^\top \bm{q}_k)(\bm{m}_j^\top \bm{q}_k),
    \end{aligned}
    \end{equation}
    where $\bm{m}_i$ denotes the $i$'th column of the centering matrix $M_n$ for all $i,j = 1,...,n$. For simplicity, we write $\sigma_{ii}$ as $\sigi^2$.
\end{lemma}

The quantities $\sigma_i^2$, $\sigma_j^2$, and $\sigma_{ij}$ denote the steady-state variances of agents $i$ and $j$ and their covariance under the delayed stochastic consensus dynamics \eqref{eq:dyn}. Through \eqref{eq:sigma_y}, they are fully determined by the Laplacian spectrum, the time-delay $\tau$, and the noise intensity $b$. The variance $\sigma_i^2$ characterizes how strongly disturbances excite disagreement at node $i$, while the covariance $\sigma_{ij}$ captures how fluctuations at agents $i$ and $j$ are coupled by the network. The correlation coefficient $\rho_{ij} = \sigma_{ij}/(\sigma_i\sigma_j)$ therefore quantifies the statistical channel through which failures propagate, and directly governs the magnitude of cascading risk.

\subsection{Risk Measures}      \label{sec:risk_measure}
To quantify the severity of undesirable fluctuations in network observables, we employ Value-at-Risk ($\VAR$) and Average Value-at-Risk ($\AVAR$) \cite{Follmer2016,rockafellar2002conditional,sarykalin2008value}. Let $y: \Omega \to \R$ be a random variable in the probability space $(\Omega, \mathcal{F}, \mathbb{P})$, and define an unsafe set $C \subset \R$ representing critical deviations, e.g., fail to reach consensus. The event $\{ y(\omega) \in C \}$ captures the occurrence of such undesirable states.

To characterize external neighborhoods of $C$, we consider a family of nested level sets $\{ C_\delta \}_{\delta \in [0, \infty]}$ satisfying, for any sequence $\{\delta_n\}_{n=1}^\infty$ with $\lim_{n \rightarrow \infty} \delta_n \rightarrow \infty$,
\begin{equation} \label{eq:level_set_prop}
    \begin{aligned}
        \text{(i)}~ C_{\delta_1} \subset C_{\delta_2} \quad \text{for } \delta_1 > \delta_2, ~~
        &\text{(ii)}~ \bigcap_{n=1}^\infty C_{\delta_n} = \lim_{n \to \infty} C_{\delta_n} = C.
    \end{aligned}
\end{equation}

We define the right-tail $\VAR_\varepsilon$ at confidence level $\varepsilon \in (0,1)$ as:
\begin{equation*}
    \mathfrak{R}_\varepsilon := \inf \left\{ z \in \R \mid \mathbb{P}(y > z) < \varepsilon \right\},
\end{equation*}
and the corresponding $\AVAR_\varepsilon$ as the expected value conditional on this upper tail:
\begin{equation*}
    \mathfrak{A}_\varepsilon := \mathbb{E} \left[ y \mid y > \VAR_\varepsilon \right].
\end{equation*}

To relate these metrics to the level sets, we define the following representation of $\AVAR_\varepsilon$ in terms of the parameter $\delta$:
\begin{equation*}
    \mathcal{A}_\varepsilon := \sup \left\{ \delta \geq 0 \mid \AVAR_\varepsilon \in C_\delta \right\},
\end{equation*}
which quantifies how deeply the tail distribution penetrates the alarm zone. A higher $\mathcal{A}_\varepsilon$ indicates greater severity of risk. The case $\mathcal{A}_\varepsilon = 0$ implies the tail remains outside $C_0$, while $\mathcal{A}_\varepsilon = \infty$ implies $\AVAR_\varepsilon \in C$. Note that while $\AVAR_\varepsilon$ is a coherent risk measure \cite{rockafellar2002conditional}, the index $\mathcal{A}_\varepsilon$ only satisfies monotonicity and subadditivity.

\section{Risk of Cascading Large Fluctuations}  \label{sec:risk}

We introduce a framework to quantify the risk of {cascading large fluctuations} in a network of multiple agents. Specifically, we assess the likelihood that an agent deviates significantly from consensus, conditioned on uncertain or partial observations of others.

Let agents be indexed by $\{1, \dots, n\}$, and define a large deviation event for agent $i$ as $\{|\bar{y}_i| > c\}$ for a threshold $c > 0$ \eqref{eq:failure}. To generalize this notion, we define a family of nested level sets $\{U_\delta\}_{\delta \in [0, \infty]}$:
\begin{equation}
    U_\delta := \left(h(\delta), \infty\right),
\end{equation}
where \( h: [0, \infty] \to [c, \infty) \) is a monotonic function satisfying the properties in \eqref{eq:level_set_prop}. These sets define alarm zones of increasing proximity to failure.
We adopt the parametric form from \cite{somarakis2023risk}:
\begin{equation}
    U_\delta := \Big(c \, \frac{\delta+1}{\delta+\alpha}, \infty\Big), \quad \alpha > 1, \label{eq:level_set}
\end{equation}
where $\alpha$ controls the rate of convergence to the unsafe region, and larger $\delta$ implies closer proximity to failure. A visual illustration of $U_\delta$, along with the associated $\VAR_\varepsilon$ and $\AVAR_\varepsilon$, is provided in Fig.~\ref{fig:level_set}. For any agent $j$ and any information set $\mathcal{O}$ describing partial or exact observations of other agents, we define the tail risk of $|\bar y_j|$ relative to $\mathcal{O}$. The $\VAR_\varepsilon$ and $\AVAR_\varepsilon$ are
\begin{align}
    \mathfrak{R}_{\varepsilon}^{\mathcal{O},j}
    &:= \inf \left\{ z \in \mathbb{R} \,\big|\, 
        \mathbb{P}\big(|\bar y_j| > z \mid \mathcal{O}\big) < \varepsilon
    \right\}, \label{eq:var_def_cond}\\[2mm]
    \mathfrak{A}_{\varepsilon}^{\mathcal{O},j}
    &:= \mathbb{E}\!\left[
        |\bar y_j| \,\big|\, |\bar y_j| > \mathfrak{R}_{\varepsilon}^{\mathcal{O},j}
    \right]. \label{eq:avar_def_cond}
\end{align}
The risk level associated with $\mathcal{O}$ is
\begin{equation}
    \mathcal{A}_{\varepsilon}^{\mathcal{O},j}
    := \sup\Big\{\, \delta \ge 0 ~\big|~
        \mathfrak{A}_{\varepsilon}^{\mathcal{O},j}
         > c\, \frac{\delta+1}{\delta+\alpha}
    \Big\}. \label{eq:risk_def_cond}
\end{equation}



\begin{figure}[t]
    \centering
    \includegraphics[width=\linewidth,height = 4.5cm]{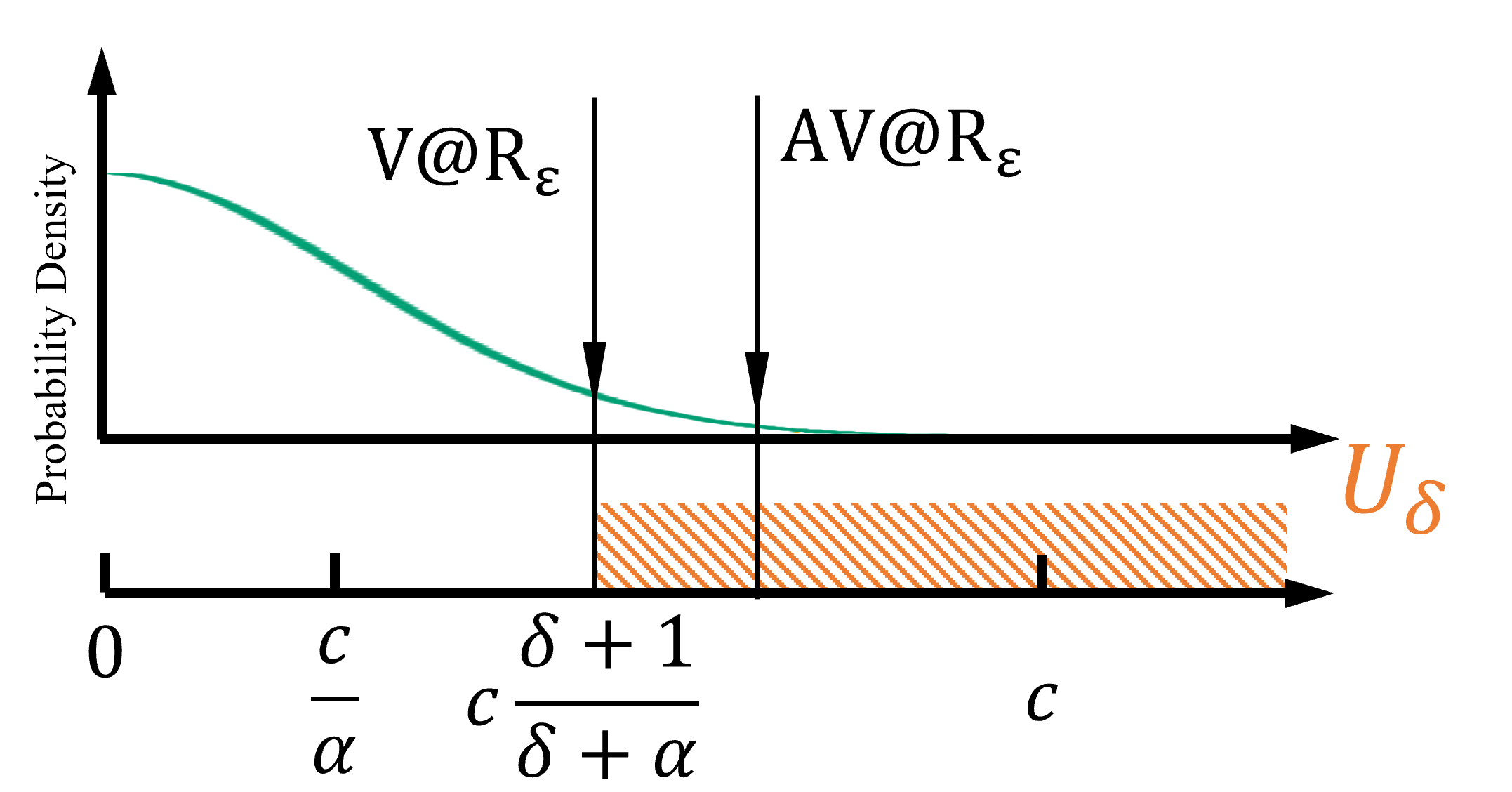}
    \caption{The concept of the risk set $U_{\delta}$, $\VAR_\varepsilon$, and $\AVAR_\varepsilon$.}
    \label{fig:level_set}
\end{figure}

\subsection{Failures Under Range-Bounded Information}

Consider the case where only partial information is available about agent $i$'s deviation from consensus, i.e., $|\bar{y}_i| \in U_{\delta^*}$, with $U_{\delta^*} = \left(c \, \frac{\delta^* + 1}{\delta^* + \alpha}, \infty\right)$. This models situations where an agent is known to be near the failure threshold $c$ but cannot be measured precisely due to sensing limitations. In such scenarios, the risk of cascading large fluctuation at the $j$'th agent can be computed using \eqref{eq:risk_def_cond} with $\mathcal{O} := \{|\bar y_i| \in U_{\delta^*}\}$, and the pair $(\bar y_i,\bar y_j)$ is jointly Gaussian at the steady-state. Let $\rho_{ij}$ denote their correlation and let $\sigma_i$, $\sigma_j$ be their standard deviations. The conditional tail probability 
$\mathbb{P}(|\bar y_j|>z \mid \mathcal{O})$
admits the representation
\begin{equation}
    \mathbb{P}(|\bar y_j|>z \mid \mathcal{O})
    = \Theta_{-}(z) + \Theta_{+}(z),
\end{equation}
where $\Theta_{\pm}$ is given by
\begin{equation}\label{eq:Theta_pm}
\scalebox{0.9}{$
\Theta_\pm (z) = \int_{\mp \infty}^{\mp z} \pm \exp\!\left(-\frac{1}{2}\left(\frac{x}{\sigma_j}\right)^2\right) \left(1 - \frac{1}{2}\Psi_{-}(x) + \frac{1}{2}\Psi_{+}(x)\right) \, \mathrm{d}x.
$}
\end{equation}
and
\begin{equation}\label{eq:Psi_pm}
\scalebox{1}{$
\Psi_{\pm}(x)=\textup{erf}\!\left(\frac{1}{\sqrt{2(1-\rho_{ij}^2)}}\left(\mp\frac{c(\delta^*+1)}{\sigi(\delta^*+\alpha)}-\frac{\rho_{ij}x}{\sigma_j}\right)\right).
$}
\end{equation}
The mapping $z \mapsto \Pro(|\bar y_j|>z \mid |\bar y_i|\in U_{\delta^*}) = \Theta_{-}(z)+\Theta_{+}(z)$ is continuous and strictly decreasing on $[0,\infty)$, with value $1$ at $z=0$ and limit $0$ as $z\to\infty$. For any $\varepsilon\in(0,1)$, the equation $\Theta_{-}(z)+\Theta_{+}(z)=\varepsilon$ admits a unique solution, denoted $\mathfrak{R}^{i,j}_\varepsilon$, which depends parametrically on $\sigma_i,\sigma_j,$ and $\rho_{ij}$. 
The corresponding conditional $\AVAR_\varepsilon$ \eqref{eq:avar_def_cond} is
\begin{equation} \label{eq:range_bounded_avar}
        \mathfrak{A}^{i,j}_{\varepsilon}
        =
        \frac{\int_{|\bar{y}_j| > \mathfrak{R}^{i,j}_{\varepsilon}} \int_{|\bar{y}_i| \in U_{\delta^*}} |\bar{y}_j|\, h(\bar{y}_{i},\bar{y}_{j}) \textup{d} \bar{y}_{i} \textup{d} \bar{y}_{j}}{2\pi\sigma_i \sigma_j \sqrt{1-\rho_{ij}^2}\,\varepsilon \left[1-\erf\!\left(\frac{c(\delta^*+1)}{\sqrt{2}\sigma_i(\delta^*+\alpha)}\right)\right]},
\end{equation}
where
\begin{align*}
    \scalebox{0.9}{$
        h(\bar{y}_{i},\bar{y}_{j}) =\exp\left(-\frac{1}{2(1-\rho^2_{ij})} \left[(1-\rho_{ij}^2) \left(\frac{\bar{y}_j}{\sigj} \right)^2 + \left(\frac{\bar{y}_i}{\sigi} - \rho_{ij}\frac{\bar{y}_j}{\sigj} \right)^2 \right] \right).
    $}
\end{align*}

\begin{theorem}\label{thm:gen_cas_ren_risk}
    Suppose that the consensus network \eqref{eq:dyn} reaches the steady-state and the $i$'th agent is close to fail the $c$-consensus with $|\bar{y}_i| \in U_{\delta^*}$. Then, the risk of cascading large fluctuation at the $j$'th agent is 
    \[
        \cass := \begin{cases}
                0 &\text{if }~ \mathfrak{A}^{i,j}_{\varepsilon} \leq \frac{c}{\alpha}\\
                \frac{\alpha \, \mathfrak{A}^{i,j}_{\varepsilon} - c}{c - \mathfrak{A}^{i,j}_{\varepsilon}} &\text{if }~ \mathfrak{A}^{i,j}_{\varepsilon} \in  \left(\frac{c}{\alpha},c \right)\\
                \infty &\text{if }~ \mathfrak{A}^{i,j}_{\varepsilon} \geq c
                \end{cases},
    \]
    where $\mathfrak{A}_{\varepsilon}^{i,j}$ is computed as in \eqref{eq:range_bounded_avar}.
\end{theorem}

The above result provides a closed form expression for evaluating the risk of cascading large fluctuations when an agent is dangerously close to failure, i.e., $|\bar{y}_i| \in U_{\delta^*}$. This result can be used to update the existing risk evaluation framework \cite{Somarakis2019g,liu2022risk} when the measurement of the system is vague. 
While the result is in closed form, the presence of nested integrals complicates its practical evaluation. In the remainder of the paper, we focus on more structured cases that enable further analytical insights.

\subsection{Cascading Risk with Partial Network Snapshots}

We refine the analysis of cascading failures by incorporating partial snapshots of agent states. These auxiliary observations offer additional context to assess how specific agent deviations influence the failure risk of others. Unlike prior formulations that assume observed agents have already failed to reach the $c$-consensus \cite{liu2022risk,liu2023cascading}, i.e., $|\bar{y}_i| > c$, our framework generalizes to arbitrary observations, including non-failure states.

To this end, we consider the risk of a cascading large fluctuation at agent $j \notin \mathcal{I}_m$, given exact observations of a subset of agents indexed by $\mathcal{I}_m = \{i_1, \cdots, i_m\} \subset \{1, \dots, n\}$ with $m < n$. Let the observed values be $\bm{y}_f = [y_{f_1}, \dots, y_{f_m}]^\top \in \mathbb{R}^m$, corresponding to the vector of random variables $\bm{\bar{y}}_{\mathcal{I}_m} = [\bar{y}_{i_1}, \dots, \bar{y}_{i_m}]^\top$. We aim to quantify the risk that agent $j$ fails to reach the $c$-consensus, conditioned on this partial snapshot. Then, the event of the cascading large fluctuation is defined as 
\begin{equation*}
    \mathcal{O} := \left\{ |\bar{y}_{j}| \in U_{\delta} \, \big| \, \bar{\bm{y}}_{\mathcal{I}
    _m}= \bm{y}_f \right\},
\end{equation*}
where $\bm{y}_f$ denotes the observed steady-state realization. The corresponding conditional $\VAR_\varepsilon$ \eqref{eq:var_def_cond} is:
\begin{equation} \label{eq:mul_var}
    \mathfrak{R}^{\mathcal{I}_m,j}_\varepsilon := \inf \left\{ z \, \Big| \, \mathbb{P}\{|\bar{y}_j| > z \,\big|\, \bar{\bm{y}}_{\mathcal{I}
    _m}= \bm{y}_f \} < \varepsilon \right\},
\end{equation}
and the conditional $\AVAR_\varepsilon$ \eqref{eq:avar_def_cond} is:
\begin{align}\label{eq:mul_avar}
    \casa = \E \left[ |\bar{y}_{j}| \, \Big |\, |\bar{y}_{j}| > \mathfrak{R}^{\mathcal{I}_m,j}_\varepsilon \wedge \bar{\bm{y}}_{\mathcal{I}
    _m}= \bm{y}_f \right].
\end{align}
The risk of cascading failures $\casm$ is then evaluated by applying these quantities to the level-set definition in \eqref{eq:risk_def_cond}.

To evaluate the terms in \eqref{eq:mul_var}, we consider the conditional distribution of $\bar{y}_j$ given partial observations $\bar{\bm{y}}_{\mathcal{I}_m} = \bm{y}_f$. Let us define the $(m+1) \times (m+1)$ block covariance matrix
\begin{equation}    \label{eq:block_cov}
    \tilde{\Sigma} = \begin{bmatrix}\,
       \tilde{\Sigma}_{11} & \tilde{\Sigma}_{12} \\
       \tilde{\Sigma}_{21} & \tilde{\Sigma}_{22} \,
    \end{bmatrix},
\end{equation}
where $\tilde{\Sigma}_{11} = \sigma_{j}^2$, $\tilde{\Sigma}_{12} = \tilde{\Sigma}_{21}^\top = [\sigma_{j i_1}, \dots, \sigma_{j i_m}]$, and $\tilde{\Sigma}_{22} = [\sigma_{k_1 k_2}]_{k_1,k_2 \in \mathcal{I}_m} \in \mathbb{R}^{m \times m}$. The terms $\sigma_{ij}$ are computed using \eqref{eq:sigma_y}. This structure enables analytical characterization of the conditional statistics of $\bar{y}_j$ given the observed values $\bm{y}_f$.

\begin{lemma}   \label{lem:multi_conditional_prob}
    Suppose the system \eqref{eq:dyn} reaches steady-state. Then, the conditional distribution of $\bar{y}_j$ given $\bm{\bar{y}}_{\mathcal{I}_m} = \bm{y}_f$ is Gaussian, i.e.,
    $
    \bar{y}_j \mid \bm{\bar{y}}_{\mathcal{I}_m} = \bm{y}_f \sim \mathcal{N}(\tilde{\mu}, \tilde{\sigma}^2),
    $
    where
    \begin{equation} \label{eq:lem_mul_cond}
       \tilde{\mu} = \tilde{\Sigma}_{12} \tilde{\Sigma}_{22}^{-1} \bm{y}_f, 
       \qquad
       \tilde{\sigma}^2 = \tilde{\Sigma}_{11} - \tilde{\Sigma}_{12} \tilde{\Sigma}_{22}^{-1} \tilde{\Sigma}_{21},
    \end{equation}
    and the sub-blocks $\tilde{\Sigma}_{11}, \tilde{\Sigma}_{12}, \tilde{\Sigma}_{21}, \tilde{\Sigma}_{22}$ are defined in \eqref{eq:block_cov}.
\end{lemma}

The above lemma enables closed-form computation of the conditional distribution of agent $j$ given partial observations from agents indexed by $\mathcal{I}_m$. Using this, we can now derive the corresponding risk of cascading large fluctuations.

\begin{theorem} \label{thm:mul_cas_ren_risk}
    Suppose the system \eqref{eq:dyn} reaches steady state, and agents indexed by $\mathcal{I}_m$ are observed at $\bar{\bm{y}}_{\mathcal{I}_m} = \bm{y}_f$. Then, the risk of cascading large fluctuation for agent $j \notin \mathcal{I}_m$ is given by:
    \begin{equation} \label{eq:casa}
        \casm := \begin{cases}
                0 &\text{if }~ \casa \leq \frac{c}{\alpha}\\
                \frac{\alpha \, \casa - c}{c - \casa} &\text{if }~ \casa \in  \left(\frac{c}{\alpha},c \right)\\
                \infty &\text{if }~ \casa \geq c
                \end{cases},
    \end{equation}
    where 
    \begin{equation*}
        \scalebox{0.85}{$
            \casa = \frac{\tilde{\sigma}}{\sqrt{2\pi}\varepsilon} \left[e^{-\frac{(\gamma + \tilde{\mu})^2}{2\tilde{\sigma}^2}} + e^{-\frac{(\gamma - \tilde{\mu})^2}{2\tilde{\sigma}^2}} + \frac{\sqrt{\pi}\tilde{\mu}}{\sqrt{2}\tilde{\sigma}}\left( \textup{erf}\left(\frac{\gamma + \tilde{\mu}}{\sqrt{2}\tilde{\sigma}}\right) - \textup{erf}\left(\frac{\gamma - \tilde{\mu}}{\sqrt{2}\tilde{\sigma}}\right)\right)\right],
        $}
    \end{equation*}
    and $\gamma$ is the unique solution of $\erf\left(\frac{\gamma-\tilde{\mu}}{\sqrt{2}\tilde{\sigma}}\right) + \erf\left(\frac{\gamma+\tilde{\mu}}{\sqrt{2}\tilde{\sigma}}\right) = 2(1-\varepsilon)$. The terms $\tilde{\mu}$ and $\tilde{\sigma}^2$ are as defined in Lemma \ref{lem:multi_conditional_prob}.
\end{theorem}

The three cases in \eqref{eq:casa} represent qualitatively distinct risk profiles. The case $\casm = 0$ indicates the scenario in which the $\AVAR_\varepsilon$ of the $j$'th agent failing to reach $c-$consensus is always less than $\frac{c}{\alpha}$, which commonly corresponds to a low confidence level or the conditional distribution of $\bar{y}_j$ concentrated away from the alarm zone $U_0$. The case of $\casm = \infty$ indicates that the $\AVAR_\varepsilon$ of the agent $j$ to be found inside the unsafe set $U_{\infty} = U$. In last case, the risk of cascading large fluctuation obtains a positive real value, and a higher value of $\casm$ indicates a higher chance that the cascading failure will occur in the system \eqref{eq:dyn}. The network-wide risk of cascading failure profile can be compactly expressed as a vector:
\begin{equation*}
    \bm{\mathcal{A}}^{\mathcal{I}_m}_{\varepsilon} = 
    \big[\mathcal{A}^{\mathcal{I}_m,1}_{\varepsilon}, \dots, \mathcal{A}^{\mathcal{I}_m,n}_{\varepsilon} \big]^\top,
\end{equation*}
where $\mathcal{A}^{\mathcal{I}_m,j}_{\varepsilon} = 0$ if $j \in \mathcal{I}_m$. When $m = 1$, this result reduces to the case analyzed in \cite{liu2022risk}.

\section{Risk of Cascading Large Fluctuations under Special Graph Topologies}
\label{sec:special-graphs}

The communication graph topology plays a key role in how time-delays and disturbances propagate through the network. This section analyzes cascading failure risk under several canonical topologies, highlighting how structural features impact the network's vulnerability to large deviations.

\subsection{The Complete Graph}
Consider a network with an unweighted complete communication graph. The Laplacian matrix has eigenvalues $\lambda_1 = 0$ and $\lambda_j = n$ for $j = 2,\dots,n$. The steady-state statistics of the network observables $\bar{\bm{y}}$ admit a closed-form expression as follows.

\begin{lemma} \label{lem:sig_complete}
    For a network \eqref{eq:dyn} with complete graph topology at steady-state, the observable satisfies $\bm{\bar{y}} \sim \N(0, \Sigma)$, where the covariance matrix $\Sigma$ has entries
    \[
        \sigma_{ij} = 
        \begin{cases}
            \frac{n-1}{2n^2} \frac{\cos (n \tau) \, b^2}{1 - \sin (n\tau)}, &\text{if } i = j\\[5pt]
            -\frac{1}{2n^2} \frac{\cos (n \tau) \, b^2}{1 - \sin (n\tau)}, &\text{if } i \neq j
        \end{cases}
    \]
    for all $i, j = 1,\dots,n$.
\end{lemma}

Given observations $\bar{\bm{y}}_{\mathcal{I}_m} = \bm{y}_f$ from agents indexed by $\mathcal{I}_m$, the conditional distribution of $\bar{y}_j$ is characterized below.

\begin{lemma} \label{lem:complete_mu_sig_cond}
The conditional distribution of $\bar{y}_j \, | \, \bm{\bar{y}}_{\mathcal{I}_m} = \bm{y}_{f}$ follows $\N(\tilde{\mu},\tilde {\sigma}^2)$, in which
    \begin{equation}    \label{eq:cond_complete}
        \tilde{\mu} = \frac{- \mathbf{1}_{m}^\top \bm{y}_f}{n-m}~ \text{, and}~\tilde{\sigma}^2 = \sigj^2 \left({1 - \frac{m}{(n-1)(n-m)}} \right).   
    \end{equation}
\end{lemma}

Applying Lemma \ref{lem:complete_mu_sig_cond} with Theorem \ref{thm:mul_cas_ren_risk} yields a closed-form expression for the risk of cascading failure $\casm$ under complete graph topology. Notably, the conditional statistics and resulting risk $\casm$ are invariant to the location of failed agents in $\mathcal{I}_m$, as confirmed by the numerical results in Fig.~\ref{fig:risk_num}.

\subsection{The Star Graph}
Consider a network with a star communication topology, where agent $n$ is the central node and agents $1,\dots,n-1$ lie on the periphery. The Laplacian eigenvalues are $\lambda_1 = 0$, $\lambda_j = 1$ for $j = 2,\dots,n-1$, and $\lambda_n = n$. While the star graph has the same sparsity as a path or 1-cycle graph, it remains connected even if some peripheral agents are disconnected. For notational convenience, define
\begin{equation} \label{eq:g}
    g(x) = \frac{\cos{x}}{1 - \sin{x}}.
\end{equation}
The steady-state covariance of the observables $\bar{\bm{y}}$ is given below.

\begin{lemma} \label{lem:sig_star}
    For a network \eqref{eq:dyn} with star topology at steady-state, the observable satisfies $\bm{\bar{y}} \sim \N(0, \Sigma)$, where the covariance matrix $\Sigma$ has the following structure:
    \[
        \sigma_{ij} = 
        \begin{cases}
            \frac{b^2}{2n(n-1)}\left(n(n-2)g(\tau) + \frac{1}{n}g(n\tau)\right), & \text{if } i = j \\[5pt]
            \frac{b^2}{2n(n-1)}\left(-n g(\tau) + \frac{1}{n}g(n\tau)\right), & \text{if } i \neq j
        \end{cases}
    \]
    for $i,j = 1,\dots,n-1$, and
    \[
        \sigma_{in} =
        \begin{cases}
            \frac{b^2(n-1)}{2n^2} g(n\tau), & \text{if } i = n \\[3pt]
            -\frac{b^2}{2n^2} g(n\tau), & \text{if } i \neq n.
        \end{cases}
    \]
\end{lemma}

The conditional distribution of $\bar{y}_j$ given partial observations $\bar{\bm{y}}_{\mathcal{I}_m} = \bm{y}_f$ depends on the location of the failed agents as detailed below.

\begin{lemma} \label{lem:star_mu_sig}
    The conditional distribution $\bar{y}_j \,|\, \bar{\bm{y}}_{\mathcal{I}_m} = \bm{y}_f$ follows a normal distribution $\N(\tilde{\mu}, \tilde{\sigma}^2)$, with the following cases:
    
    \noindent \underline{Case (i):} All $m$ failures are on the periphery:
    \[
        \tilde{\mu} =
        \begin{cases}
            \frac{-(n - 1)g(n\tau)}{\Delta} \mathbf{1}_m^\top \bm{y}_f, & \text{if } j = n \\[3pt]
            \frac{-n^2 g(\tau) + g(n\tau)}{\Delta} \mathbf{1}_m^\top \bm{y}_f, & \text{if } j \neq n
        \end{cases}, 
    \]
    and
    \[
        \tilde{\sigma}^2 =
        \begin{cases}
            \sigma_n^2 \frac{n^2(n - m - 1)g(\tau)}{\Delta}, & \text{if } j = n \\[5pt]
            \frac{b^2 g(\tau)}{2} \left(1 + \frac{g(n\tau)}{\Delta}\right), & \text{if } j \neq n
        \end{cases},
    \]
    where $\Delta = n^2(n - m - 1)g(\tau) + m g(n\tau)$ and $\mathbf{1}_m \in \R^m$ is the all-ones vector.

    \noindent \underline{Case (ii):} One failure is at the center ($i = n$), and $m-1$ are on the periphery:
    \[
        \tilde{\mu} = \frac{- \mathbf{1}_m^\top \bm{y}_f}{n - m}, \qquad
        \tilde{\sigma}^2 = \frac{1}{2} b^2 g(\tau) \left(1 - \frac{1}{n - m}\right),
    \]
    where $\bm{y}_f(m)$, the last element of $\bm{y}_f$, corresponds to the central agent's observable.
\end{lemma}
These results can be used in conjunction with Theorem \ref{thm:mul_cas_ren_risk} to compute the closed-form expression for $\casm$ under the star topology.

While other classical graphs such as paths or cycles admit explicit Laplacian spectra, the delay-modified spectral sums arising in the risk expressions do not simplify analytically, preventing explicit scaling characterizations.

\section{Efficient Single-Step Update Law for Calculating Risk of Cascading Failures }  \label{sec:update}

We consider the scenario where $m$ agents indexed by $\mathcal{I}_m$ are already in failure states, and we aim to update the conditional distribution $\bar{y}_{j} \,|\, \bar{\bm{y}}_{\mathcal{I}_m} 
= \bm{y}_f$ when an additional failure is detected at agent $k \notin \mathcal{I}_m$, $k \neq j$. Instead of recomputing the full conditional distribution via Lemma \ref{lem:multi_conditional_prob}, which involves inverting an $(m+1)\times(m+1)$ covariance submatrix, we derive an efficient update law that incrementally adjusts the conditional statistics.

To this end, define:
\begin{equation}    \label{eq:j_and_k}
    \begin{aligned}
        &\tilde{\mu}_j = \tilde{\Sigma}_{12}(j)\tilde{\Sigma}_{22}^{-1}\bm{y}_f,
        ~~ \tilde{\sigma}^2_j = \sigj^2 -  \tilde{\Sigma}_{12}(j)\tilde{\Sigma}_{22}^{-1}\tilde{\Sigma}_{21}(j),\\
        &\tilde{\mu}_k = \tilde{\Sigma}_{12}(k)\tilde{\Sigma}_{22}^{-1}\bm{y}_f,
        ~~ \tilde{\sigma}^2_k = \sigma_k^2 -  \tilde{\Sigma}_{12}(k)\tilde{\Sigma}_{22}^{-1}\tilde{\Sigma}_{21}(k),
    \end{aligned}
\end{equation}
where $\tilde{\Sigma}_{12}(k) = \tilde{\Sigma}_{21}^\top(k)$ is the cross-covariance between agent $k$ and the observed failures $\mathcal{I}_m$; similar notation holds for agent $j$. The matrix $\tilde{\Sigma}_{22}^{-1}$ corresponds to the inverse covariance of the observed agents and is reused across updates.

\begin{theorem} \label{thm:conditional_prob_update}
Suppose $\bar{y}_{j} \,|\, \bar{\bm{y}}_{\mathcal{I}_m} = \bm{y}_f \sim \N(\tilde{\mu}_j, \tilde{\sigma}_j^2)$, where $\mathcal{I}_m$ indexes the current $m$ failures. When a new failure is observed at agent $k \notin \mathcal{I}_m$, $k \neq j$, with measurement $\bar{y}_k = y_{f_k}$ satisfying $|y_{f_k}| > c$, the updated conditional distribution is given by $\N(\hat{\mu}, \hat{\sigma}^2)$, where
\begin{align} \label{eq:update}
    \hat{\mu} = \tilde{\mu}_j - \frac{\tilde{\sigma}_{jk}}{\tilde{\sigma}_k^2}(\tilde{\mu}_k - y_{f_k}), ~ \text{and} ~~
    \hat{\sigma}^2 = \tilde{\sigma}_j^2 - \frac{\tilde{\sigma}_{jk}^2}{\tilde{\sigma}_k^2}, \nonumber
\end{align}
and the cross-covariance term $\tilde{\sigma}_{jk}$ is given by
\[
    \tilde{\sigma}_{jk} = \sigma_{jk} - \tilde{\Sigma}_{12}(j)\tilde{\Sigma}_{22}^{-1}\tilde{\Sigma}_{21}(k).
\]
All terms are computed using \eqref{eq:j_and_k} and the precomputed $\tilde{\Sigma}_{22}^{-1}$ from the current failure set.
\end{theorem}

\begin{figure}
    \centering
    \includegraphics[width=1\linewidth]{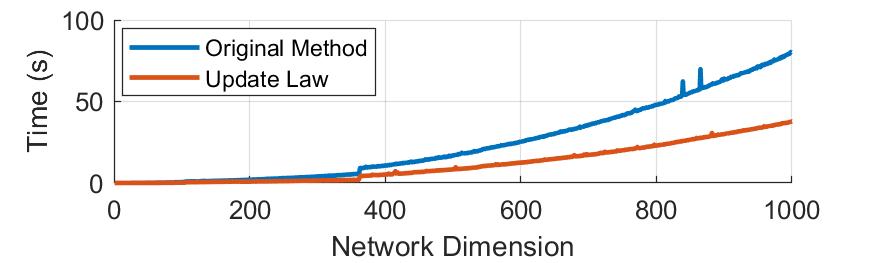}
    \caption{Computation time with various dimensions of the network.}
    \label{fig:computation_time_compare}
\end{figure}

This update rule provides a fast and scalable mechanism to propagate cascading failure risk as new agent failures are detected. By avoiding reconstruction of the full conditional distribution, the computational cost is significantly reduced (see Fig.~\ref{fig:computation_time_compare}).

\section{Time-Delay Induced Fundamental Limits on Cascading Risk} \label{specialgraph}

In many engineering systems, communication delays and external disturbances are intrinsic and not directly controllable. As a result, mitigating the risk of cascading failures must rely on modifying the network topology, specifically, by adjusting the feedback gains on communication links. This section characterizes fundamental performance limits on the risk of large deviations induced by time-delay, under general communication graph topologies.

\subsection{Fundamental Limits and the Lower Bound of the Best Achievable Risk}

In the presence of the communication time-delay, there exists a time-delay-induced fundamental limits on the elements of the covariance of the steady-state observable \eqref{eq:sigma_y}. To reveal this, the following limits on the $f$ function, which appears in \eqref{eq:sigma_y}, is introduced in order to develop the limits on the steady-state covariance.

\begin{lemma}   \label{lem:f_lower_bound}
    The function $f (x) = \frac {1}{2x} \frac {\cos(x)}{1 - \sin(x)}$ obtains a local minimum $\underline{f}$ for $x \in (0, \frac{\pi}{2})$, where 
    $$
    \underline{f}:= \inf_{x\in (0,\frac{\pi}{2})} f(x) = \inf_{x\in (0,\frac{\pi}{2})} \frac {1}{2x} \frac {\cos(x)}{1 - \sin(x)}  \approx 1.5319.
    $$ 
\end{lemma}

This property of $f(x)$ yields the following bounds on the diagonal and off-diagonal elements of the covariance matrix $\Sigma$.

\begin{theorem}     \label{thm:sigma_bound}
Suppose the network \eqref{eq:dyn} reaches steady-state. Then, the entries of the covariance matrix $\Sigma$ of $\bar{\bm{y}}$ satisfy:
\begin{equation*}
    \begin{cases}
        \frac{(n-1) b^2 \tau}{n} \underline{f} \leq \sigi^2 \leq \frac{(n-1) b^2 \tau}{n} \bar{f} & \text{for all } i, \\[4pt]
        \frac{(n-2)b^2 \tau}{2n} \underline{f} - \frac{b^2 \tau}{2} \bar{f} \leq \sigma_{ij} \leq \frac{(n-2)b^2 \tau}{2n} \bar{f} - \frac{b^2 \tau}{2} \underline{f} & \text{for all } i \neq j,
    \end{cases}
\end{equation*}
where 
\[
\bar{f} := \max\left\{ f(\lambda_2 \tau), f(\lambda_n \tau) \right\}.
\]
\end{theorem}

We note that the bounds in \Cref{thm:sigma_bound} are worst-case envelopes obtained via spectral extremization of $f(\lambda_i\tau)$. Their conservativeness is evaluated numerically in \Cref{sec:case-study}, where the analytical limits are compared against exact steady-state covariance values across representative graph topologies. The above result holds for any communication graph satisfying the stability condition in Assumption~\ref{asp:stable}. 

To obtain a uniform and practically meaningful bound, we restrict attention to the compact interval $\bar{S} := [10^{-3}, \pi/2 - 10^{-3}]$, over which $f(x)$ is continuous and therefore attains its maximum. This interval removes an arbitrarily small neighborhood of the stability boundary $\lambda_i \tau = \pi/2$, where $f(x)$ diverges, thereby enforcing a finite stability margin consistent with practical implementations. Over this domain, the maximum value of $f(x)$ is approximated as
\[
\bar{f} := \sup_{x \in \bar{S}} \frac{1}{2x} \cdot \frac{\cos(x)}{1 - \sin(x)} \approx 6.3666 \times 10^3.
\]

Substituting this uniform bound into \Cref{thm:sigma_bound} yields topology-independent covariance envelopes for networks whose spectra satisfy $\lambda_i \tau \in \bar{S}$. These uniform bounds will be used below to derive a fundamental lower limit on the best achievable cascading risk.

The above covariance bounds can be leveraged to derive a fundamental lower bound on the best achievable risk of cascading failures in the network. In what follows, we focus on the case of a single initial failure, i.e., $m = 1$.

\begin{theorem} \label{thm:risk_bound}
Suppose the network \eqref{eq:dyn} satisfies $\lambda_i \tau \in \bar{S}$ for all $i = 2,\dots,n$.
Define $\sigma_{\min}:=\sqrt{\frac{n-1}{n}\,b^2\tau\,\underline{f}}$, 
$\kappa_{\varepsilon}:=\big(\sqrt{2\pi}\,\varepsilon\,e^{\iota_\varepsilon^2}\big)^{-1}$, and 
$\iota_\varepsilon:=\erf^{-1}(2\varepsilon-1)$. Then, the best achievable risk of cascading failure $\mathcal{A}_{\varepsilon}^{ij}$ is lower-bounded as follows:

\noindent \textbf{Case 1:} If $\sigma_{ij} > 0$, then $\mathcal{A}_{\varepsilon}^{ij} \geq \mathcal{A}_{+}$, where
$\mathfrak{A}_{+}:=\min\{\kappa_{\varepsilon}\sigma_{\min},\sqrt{\underline{f}/\bar{f}}\,y_f\}$ and
\[
\mathcal{A}_{+} := \begin{cases}
0 &\text{if }~ \mathfrak{A}_{+} \leq \frac{c}{\alpha} \\[4pt]
\frac{\alpha \mathfrak{A}_{+} - c}{c - \mathfrak{A}_{+}} &\text{if }~ \mathfrak{A}_{+} \in \left(\frac{c}{\alpha}, c \right) \\[4pt]
\infty &\text{if }~ \mathfrak{A}_{+} \geq c
\end{cases}.
\]

\noindent \textbf{Case 2:} If $\sigma_{ij} < 0$, then $\mathcal{A}_{\varepsilon}^{ij} \geq \mathcal{A}_{-}$, where
\[
\mathcal{A}_{-} := 0.
\]

\noindent \textbf{Case 3:} If $\sigma_{ij} = 0$, \emph{apply Case~1} with
\[
\mathfrak{A}_{+}\ \mapsto\ \mathfrak{A}_{0}:=\kappa_{\varepsilon/2}\,\sigma_{\min},\]
i.e., restrict to the endpoint $s=0$ and replace $\kappa_{\varepsilon}$ by $\kappa_{\varepsilon/2}$.

\end{theorem}

The sign and magnitude of the covariance $\sigma_{ij}$ depend on the communication graph topology, which directly affects the lower bound on the best achievable risk of cascading failures. These bounds serve as fundamental performance limits and can be used to assess whether a desired safety specification is achievable through network design.

\subsection{Best Achievable Risk with Complete Graph}

When the communication graph is specified, for example as a complete graph, the previously derived bounds on the covariance and the best achievable risk of cascading failure can be made tighter and less conservative.

\begin{corollary} \label{cor:best_risk_complete}
    Consider the network \eqref{eq:dyn} with $n$ agents communicating over an unweighted complete graph. Then the fundamental lower bound on the best achievable risk of cascading large fluctuations is attained by choosing
    \begin{equation*}
        \tilde{\mu} = \frac{- y_f}{n-1}, \quad \text{and} \quad
        \tilde{\sigma} = \sqrt{\frac{(n-2) b^2\tau}{n-1} \, \underline{f}},
    \end{equation*}
    as in \eqref{eq:casa}, with the corresponding case-specific risk branches applied.
\end{corollary}

In contrast to Case~(2) in Theorem~\ref{thm:risk_bound}, where the best achievable risk can be trivially zero due to negative correlations between agents, the complete graph topology enforces symmetric and positive interactions among all nodes. As a result, it yields a nontrivial and informative lower bound on the achievable risk of cascading failure. This is particularly important from a design perspective, as trivial bounds (e.g., zero risk) offer limited utility in evaluating whether a given network can satisfy safety requirements under realistic time-delay and disturbance conditions.

\section{Case Studies} \label{sec:case-study}

\begin{figure*}[t]
    \centering
    \includegraphics[width=\linewidth]{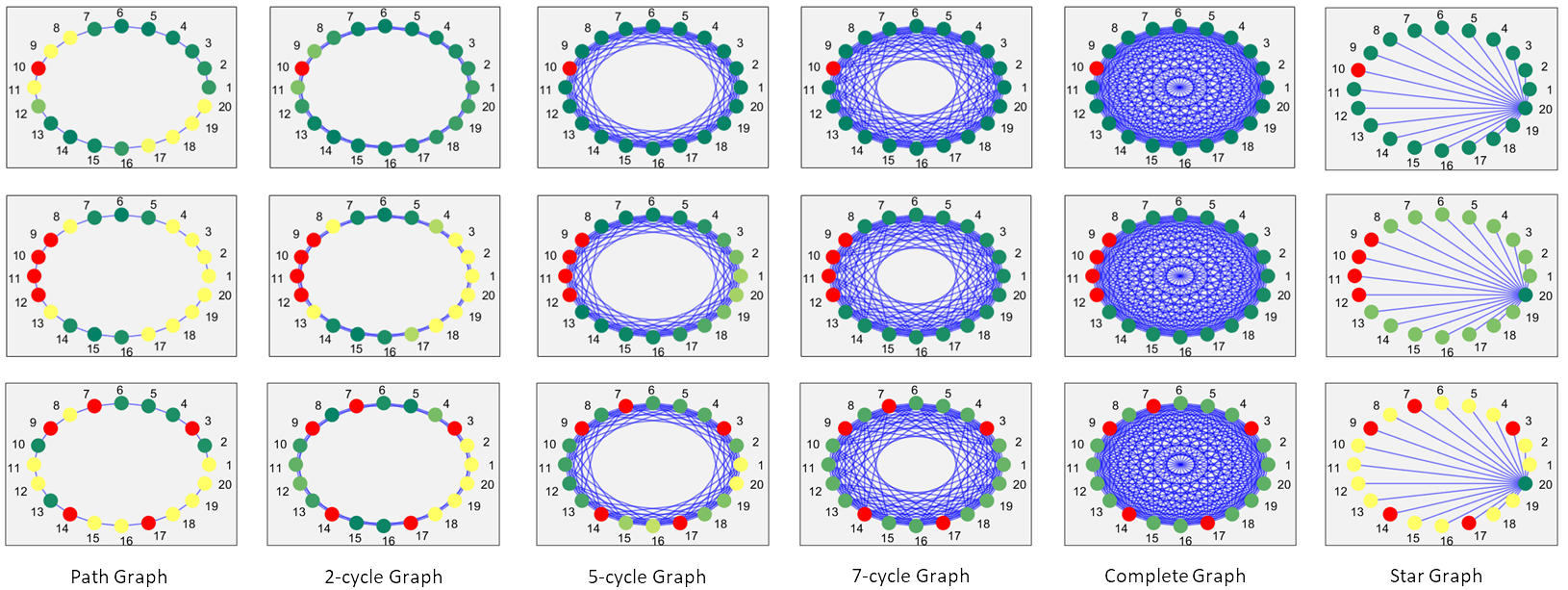}
    \caption{Network-wide risk of cascading failure profile $\rpro$ across different communication topologies. Lighter colors indicate higher risk of cascading failure $\casm$, while red nodes denote existing failures. Communication links are shown in blue.}
    \label{fig:cas_risk}
\end{figure*}

We examine the rendezvous problem governed by the stochastic consensus dynamics in \eqref{eq:dyn} under several canonical communication topologies, including the complete, path, and $p$-cycle graphs \cite{van2010graph}. In each case study, the agents indexed by $\mathcal{I}_m$ are assumed to have failed to achieve the $c$-consensus and exhibit large fluctuations characterized by $\bm{y}_f = y_f\bm{1}_m$. Unless otherwise stated, the simulation parameters are chosen as $n = 20$, $c = 4$, $\alpha = 1000$, $y_f = 4$, $b = 0.01$, $\tau = 0.05$, and $\varepsilon = 0.1$.

\subsection{Risk of Cascading Large Fluctuation}

The network-wide risk of cascading failure profile $\rpro$ is evaluated using the closed-form expressions derived in Theorem~\ref{thm:mul_cas_ren_risk} across several unweighted communication topologies. The resulting risk distributions are illustrated in Fig.~\ref{fig:cas_risk}.

\noindent \underline{Path Graph:}
Agents are arranged in a linear topology, each communicating only with its immediate neighbors. The risk of cascading failure $\casm$ is highly localized around the initially failed node and decays rapidly with increasing graph distance. When the failure occurs near the network boundary, the risk diminishes faster toward the edge and more gradually toward the interior, resulting in an asymmetric risk profile. This spatial attenuation reflects the limited diffusion of disturbances in sparsely connected graphs and matches the theoretical dependence of $\casm$ on Laplacian eigenmodes in Theorem~\ref{thm:mul_cas_ren_risk}.

\begin{figure}[t]
    \centering
    \includegraphics[width=\linewidth]{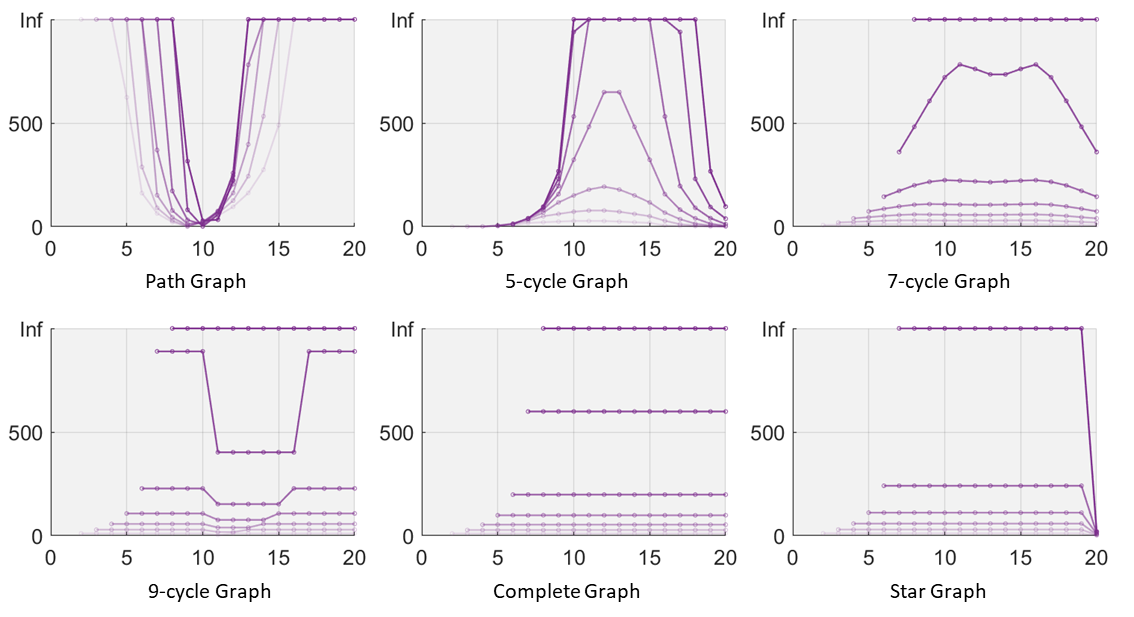}
    \caption{The risk profile with a different number of failures occurs at agent $1$ to $7$. The $y$ denotes the risk value $\casm$, the $x$ axis represents the agent's label, and a darker color denotes more existing failures.}
    \label{fig:risk_num}
\end{figure}

\vspace{0.5mm}
\noindent \underline{$p$-Cycle Graph:}
Agents are connected in a cyclic topology where each node communicates with up to $p$ nearest neighbors on each side. The resulting risk of cascading failure $\casm$ exhibits a localized peak around the failed node and decays symmetrically along the cycle. As $p$ increases, information exchange becomes denser, and the risk distribution gradually transitions toward that of a complete graph, where spatial variation vanishes. This behavior highlights the trade-off between connectivity and risk localization predicted by the spectral structure of the Laplacian.

\vspace{0.5mm}
\noindent \underline{Complete Graph:}
In the complete topology, every agent communicates with all others. Consequently, all nodes experience identical risk of cascading failure $\casm$, independent of their position in the network. The uniform risk distribution arises from the perfect symmetry of the complete graph and confirms the results in Lemmas~\ref{lem:sig_complete} and~\ref{lem:complete_mu_sig_cond}. This case serves as a limiting benchmark where topological homogeneity eliminates spatial dependence in risk propagation.

\vspace{0.5mm}
\noindent \underline{Star Graph:}
In the star topology, a single central node connects to all peripheral nodes, while peripheral agents communicate only through the center. Simulations show that the central agent experiences the highest risk of cascading failure due to its direct exposure to all disturbances, whereas the peripheral nodes share identical but lower $\casm$ values. This asymmetric pattern aligns with the theoretical characterization in Lemma~\ref{lem:star_mu_sig} and underscores the vulnerability of hub nodes in hierarchically structured networks.

\begin{figure}[t]
    \centering
    \includegraphics[width=\linewidth]{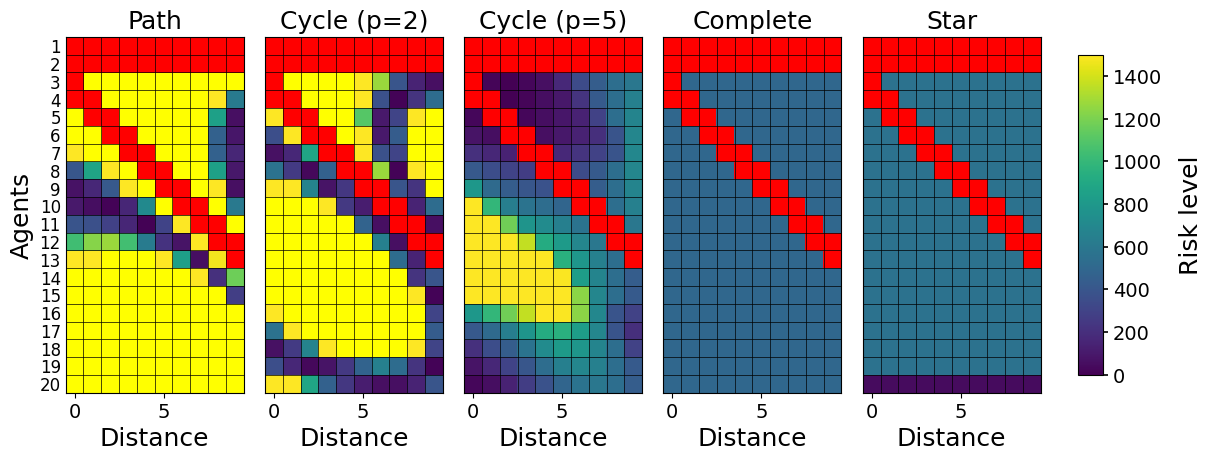}
    \caption{Cascading risk profile $\rpro$ for a fixed number of existing failures $m$ placed at different locations in the graph. Yellow indicates $ \casm = \infty$; red nodes denote existing failures.}
    \label{fig:risk_loc}
\end{figure}

\subsection{Characteristics of Existing Failures}
When multiple agents fail to maintain the $c$-consensus, both the number and spatial distribution of these failures significantly influence the overall risk landscape. In this section, we analyze two distinct characteristics of the existing failure set $\mathcal{I}_m$: (\emph{i}) the number of failed agents $m$, and (\emph{ii}) their spatial distribution in the communication graph.

\begin{figure*}[t]
    \centering
    \includegraphics[width=\linewidth]{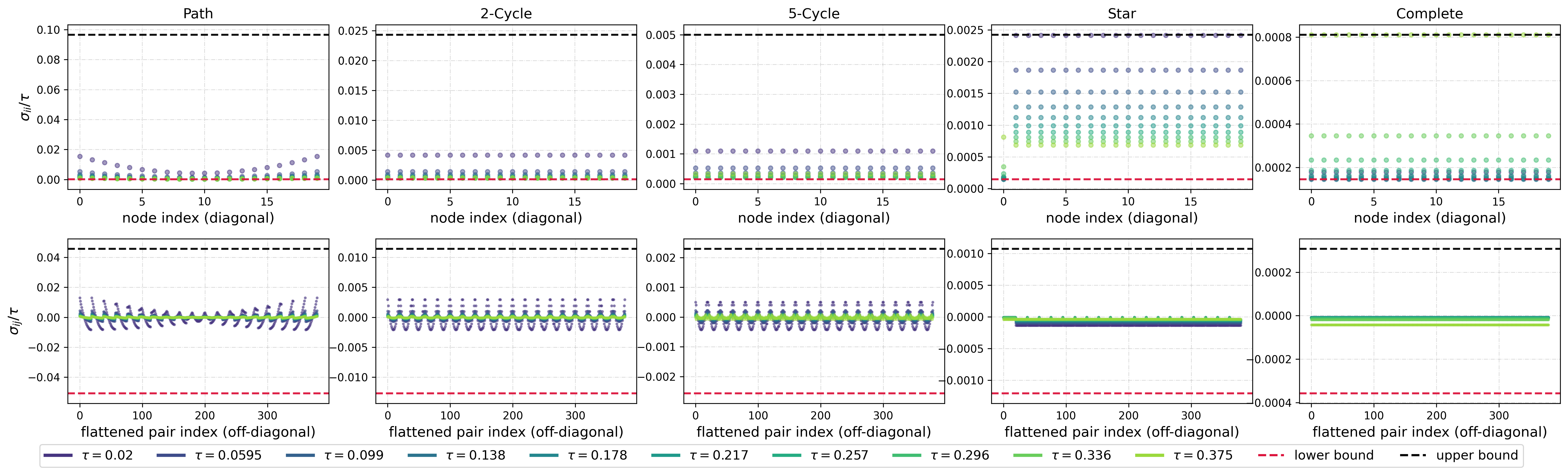}
    \caption{Evaluation of covariance bounds on selected network topologies. The top row shows diagonal pairs $(i,i)$ and the bottom row shows off-diagonal pairs $(i,j)$, each evaluated over varying time-delays $\tau$.}
    \label{fig:sigma_bound_all}
\end{figure*}

\begin{figure*}[t]
    \centering
    \includegraphics[width=\linewidth]{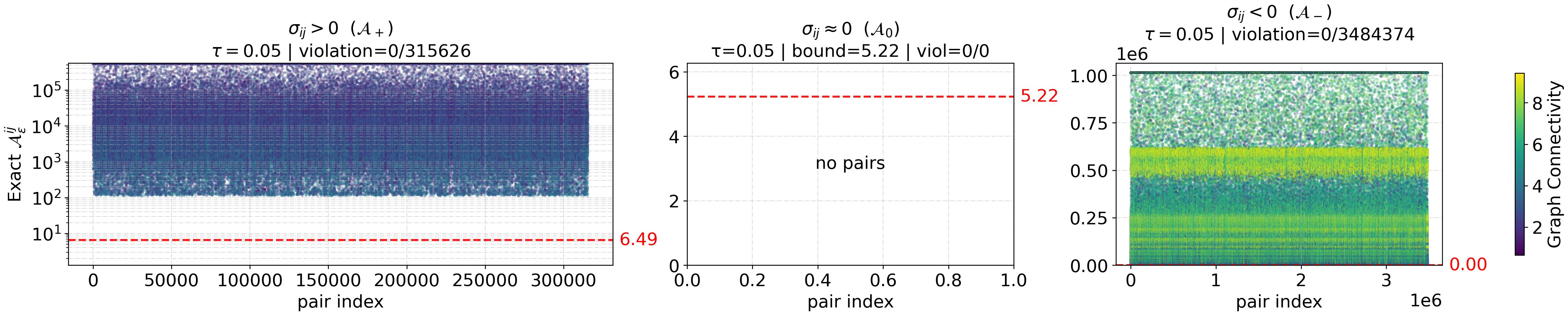}
    \caption{Empirical validation of the best-achievable risk of cascading failures $\cass$ in Theorem~\ref{thm:risk_bound} using $10^4$ randomly generated connected graphs satisfying Assumption~\ref{asp:stable}. 
    Each point represents an ordered node pair $(i,j)$, totaling $3.8\times10^6$ pairs across all graphs. The pair indices are flattened, with no cases of $\sigma_{ij}=0$ and no violations of the theoretical bound observed.}
    \label{fig:risk_bound}
\end{figure*}

\begin{figure}[t]
    \centering
    \includegraphics[width=\linewidth]{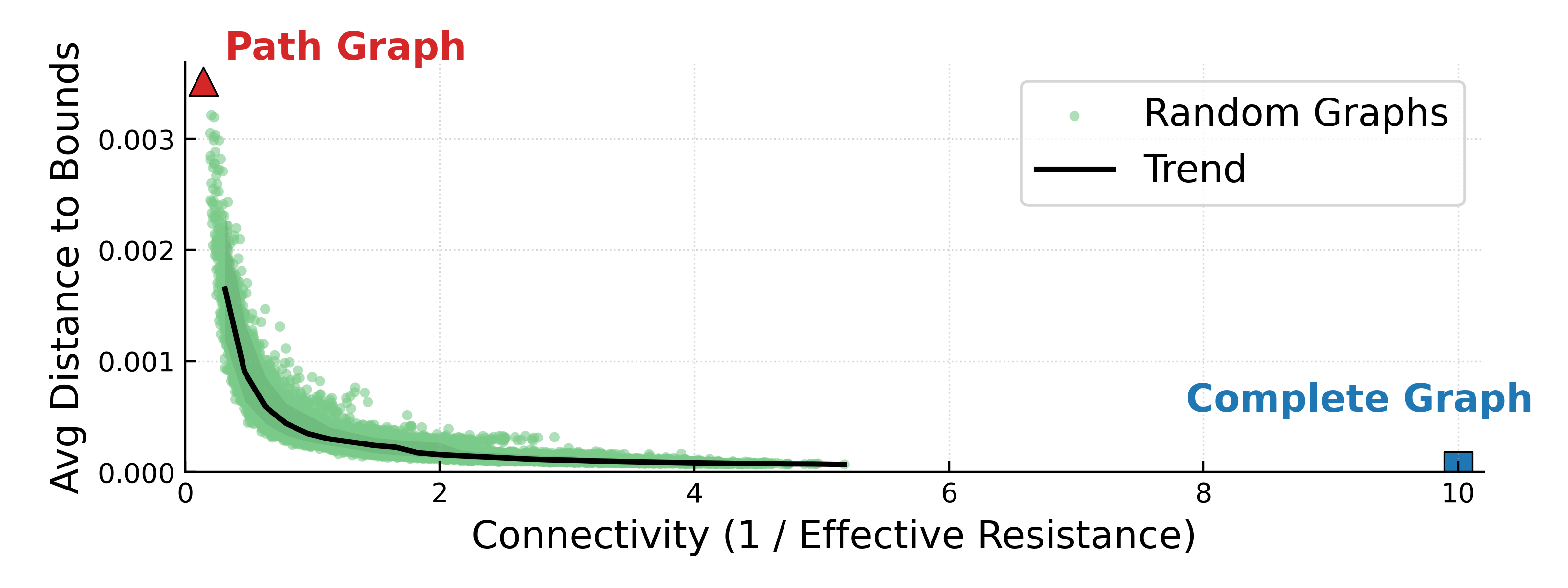}
    \caption{Average deviation of $\sigma_{ij}$ from the theoretical bounds versus graph connectivity for $n=20$. 
    Each point represents one connected graph, with the path and complete graphs shown as the extrema corresponding to the lowest and highest connectivity, respectively.}
    \label{fig:sigma_tightness_vs_connectivity}
\end{figure}

\vspace{1mm}
\noindent \underline{Number of Existing Failures:}
Figure~\ref{fig:risk_num} illustrates how the risk of cascading failure $\casm$ evolves as the number of failed agents increases. The results show clear topology-dependent patterns.  
For complete graphs, $\casm$ remains uniform across all agents regardless of $m$, confirming that perfect symmetry yields identical risk levels. In contrast, the path and $p$-cycle graphs exhibit localized and asymmetric growth of risk: as $m$ increases, the region of elevated $\casm$ broadens outward from the failure cluster, with the highest peaks forming near the boundaries or adjacent to existing failures.  
The $5$-cycle case shows a smoother and more uniform risk distribution than the path or $2$-cycle, reflecting the stronger coupling and reduced spatial localization at higher connectivity.  
A counter-intuitive observation from Fig.~\ref{fig:risk_num} is that greater connectivity does not always mitigate the risk. In the presence of time delay, tighter coupling can amplify correlations among agents, causing $\casm$ to increase near the failed nodes—consistent with the delay-induced trade-off discussed in~\cite{Somarakis2019g}.

\vspace{1mm}
\noindent \underline{Location Distribution:}
Figure~\ref{fig:risk_loc} fixes the number of existing failures $m$ and varies their spatial placement. In the path and $p$-cycle graphs, clustered failures merge their influence zones and form a single ridge of high risk of cascading failure $\casm$, whereas spatially separated failures produce multiple localized peaks whose magnitudes decay with graph distance. Increasing $p$ smooths the profile and broadens the affected region, while boundary failures in the path yield asymmetric spreading. In contrast, the complete and star graphs exhibit location-invariant behavior: for a fixed $m$, the overall $\casm$ remains unchanged regardless of where the failures occur, provided they are not at the central node in the star topology. This invariance agrees with Lemmas~\ref{lem:complete_mu_sig_cond} and~\ref{lem:star_mu_sig}, where the conditional statistics depend solely on the number of failures and topological symmetry rather than their spatial indices.

\subsection{Fundamental Limits on the Risk of Cascading Failures}

We evaluate the empirical tightness of the covariance bounds derived in Theorem~\ref{thm:sigma_bound} by computing the pairwise risk of cascading failures in several representative network topologies, as illustrated in Fig.~\ref{fig:sigma_bound_all}. The results show that the analytical covariance bounds are effective, and the gap between the empirical and theoretical values narrows as the graph becomes more connected.

Fig.~\ref{fig:sigma_tightness_vs_connectivity} quantifies this behavior by plotting the average deviation between empirical $\sigma_{ij}$ values and their theoretical limits as a function of graph connectivity, measured by the effective resistance
\begin{equation}
\label{eq:ER}
R_{\mathrm{eff}} = \frac{1}{n-1} \sum_{k=2}^{n} \frac{1}{\lambda_k},
\end{equation}
where $\lambda_2,\ldots,\lambda_n$ are the nonzero Laplacian eigenvalues. Smaller $R_{\mathrm{eff}}$ corresponds to stronger global connectivity, and the results show that the deviation decreases monotonically as $R_{\mathrm{eff}}$ decreases. Dense or expander-like networks thus exhibit tight covariance envelopes due to their concentrated spectra, whereas sparse topologies such as paths and small $p$-cycles show larger variations arising from wide spectral gaps. These findings confirm that $\bar{S} := [10^{-3}, \pi/2 - 10^{-3}]$ serves as an asymptotically sharp envelope for the feasible range of $\lambda_i\tau$ across all connected graphs satisfying Assumption~\ref{asp:stable}.

We next validate the empirical correctness of the lower bound in Theorem~\ref{thm:risk_bound}. A total of $10{,}000$ connected graphs with $n=20$ are randomly generated under Assumption~\ref{asp:stable} following the Erd\H{o}s--R\'enyi model~\cite{Erdos1959}, with edge probabilities uniformly sampled from a prescribed interval to ensure connectivity.  
The parameters are kept consistent with Sec.~\ref{sec:case-study} except that $\alpha=10{,}000$ and $c=2$. For each generated graph, we compute both the theoretical best achievable risk and the empirical risk of cascading failure $\cass$ across all node pairs. The comparison is shown in Fig.~\ref{fig:risk_bound}, where the red dashed line denotes the best achievable risk, which remains identical across all generated graphs since it depends solely on the global parameters $(\alpha,c,\epsilon)$ rather than the specific network topology.

All samples satisfy the analytical best achievable risk, confirming its universal validity across connected topologies. As graph connectivity increases, the points concentrate near the diagonal, indicating that the bound becomes tight for dense or expander-like graphs. In contrast, sparse graphs display a larger spread due to mixed signs of $\sigma_{ij}$, consistent with the three-case structure in Theorem~\ref{thm:risk_bound}.  
Beyond theoretical interest, this result provides a practical advantage for network design: the derived best achievable risk serves as a feasibility certificate, allowing one to assess whether a desired cascading-risk target can be achieved \emph{without} exhaustively enumerating all possible graph configurations. Hence, the best achievable risk acts as a universal and computationally efficient benchmark for evaluating the attainable cascading-risk level in any connected consensus network satisfying the stability condition.

\section{Conclusion}

This work presented a unified framework for quantifying cascading failures in time-delay consensus networks through the lens of the Average Value-at-Risk ($\AVAR$) measure. Building upon the stochastic consensus model for temporal rendezvous, we characterized how existing failures reshape the steady-state distribution of agent deviations and derived closed-form expressions for the resulting risk of cascading failures. The formulation captures both the marginal variances and pairwise correlations of the network observables, thereby linking the risk of secondary failures directly to the Laplacian spectrum, the communication time-delay, and the noise intensity.

Theoretical analysis established explicit lower bounds on the best-achievable risk of cascading failures that hold for any connected topology satisfying the delay stability condition. These bounds expose fundamental performance limits imposed by time-delay and graph connectivity, and they act as fast feasibility certificates for design targets without requiring exhaustive simulation across candidate graphs. Numerical studies on canonical graphs revealed distinct topological signatures of risk of cascading failure, including localization and asymmetry on paths, spatial uniformity on complete graphs, and hub dominance on stars. Large-scale experiments with $10^4$ randomly generated connected graphs confirmed that all realizations respect the analytical lower bound, which becomes tight as connectivity increases.

Overall, the proposed framework provides a systematic foundation for assessing the system's vulnerability and quantifying how existing failures amplify risk propagation in delayed multi-agent networks. Beyond analysis, a single-step incremental update rule enables efficient re-evaluation of conditional risk as new failures are observed, which substantially reduces computation time compared with recomputing from scratch. Future directions include extending the analysis to distributionally robust formulations \cite{liu2024data, pandey2023quantification, pandey2025distributionally} that capture uncertainty in noise statistics, developing adaptive control strategies to mitigate risk of cascading failure in real time, and exploring extensions to nonlinear or switching network dynamics.

\printbibliography

@book{van2010graph,
    title = {{Graph spectra for complex networks}},
    year = {2010},
    author = {Van Mieghem, Piet},
    publisher = {Cambridge University Press}
}

@book{Follmer2016,
    title = {{Stochastic Finance}},
    year = {2016},
    booktitle = {Stochastic Finance},
    author = {F{\"{o}}llmer, Hans and Schied, Alexander},
    month = {7},
    publisher = {De Gruyter}
}

@book{tong2012multivariate,
    title = {{The multivariate normal distribution}},
    year = {2012},
    author = {Tong, Yung Liang},
    publisher = {Springer Science {\&} Business Media}
}

@article{gray2006toeplitz,
    title = {{Toeplitz and circulant matrices: A review}},
    year = {2006},
    author = {Gray, Robert M},
    publisher = {now publishers inc}
}

@ARTICLE{7438924,
  author={M. {Rahnamay-Naeini} and M. M. {Hayat}},
  journal={IEEE Transactions on Smart Grid}, 
  title={Cascading Failures in Interdependent Infrastructures: An Interdependent Markov-Chain Approach}, 
  year={2016},
  volume={7},
  number={4},
  pages={1997-2006},
  }

@article{zhang2019robustness,
  title={Robustness of interdependent cyber-physical systems against cascading failures},
  author={Zhang, Y. and Ya{\u{g}}an, O.},
  journal={IEEE Transactions on Automatic Control},
  volume={65},
  number={2},
  pages={711--726},
  year={2019},
  publisher={IEEE}
}

@article{zhang2018cascading,
  title={Cascading failures in interdependent systems under a flow redistribution model},
  author={Zhang, Y. and Arenas, A. and Ya{\u{g}}an, O.},
  journal={Physical Review E},
  volume={97},
  number={2},
  pages={022307},
  year={2018},
  publisher={APS}
}

@inproceedings{ghaedsharaf2016interplay,
  title={Interplay between performance and communication delay in noisy linear consensus networks},
  author={Ghaedsharaf, Y. and Siami, M. and Somarakis, C. and Motee, N.},
  booktitle={2016 European Control Conference (ECC)},
  pages={1703--1708},
  year={2016},
  organization={IEEE}
}

@inproceedings{Somarakis2017a,
    title = {{Aggregate fluctuations in time-delay linear consensus networks: A systemic risk perspective}},
    year = {2017},
    booktitle = {Proceedings of the American Control Conference},
    author = {Somarakis, C. and Ghaedsharaf, Y. and Motee, N.}
}

@inproceedings{Somarakis2016g,
    title = {{Interplays Between Systemic Risk and Network Topology in Consensus Networks}},
    year = {2016},
    booktitle = {IFAC-PapersOnLine},
    author = {Somarakis, C. and Siami, M. and Motee, N.},
    number = {22},
    volume = {49}
}

@article{rockafellar2000optimization,
    title = {{Optimization of Conditional Value-at-Risk}},
    year = {1999},
    journal = {Portfolio The Magazine Of The Fine Arts},
    author = {Rockafellar, R. T. and Uryasev, S.},
    pages = {1--26},
    volume = {2}
}

@article{Somarakis2019g,
    title = {{Time-delay origins of fundamental tradeoffs between risk of large fluctuations and network connectivity}},
    year = {2019},
    journal = {IEEE Transactions on Automatic Control},
    author = {Somarakis, C. and Ghaedsharaf, Y. and Motee, N.},
    number = {9},
    volume = {64}
}

@INPROCEEDINGS{9683468,
  author={Liu, Guangyi and Somarakis, Christoforos and Motee, Nader},
  booktitle={2021 60th IEEE Conference on Decision and Control (CDC)}, 
  title={Risk of Cascading Failures in Time-Delayed Vehicle Platooning}, 
  year={2021},
  volume={},
  number={},
  pages={4841-4846}}

@INPROCEEDINGS{liu2022risk,
  author={Liu, Guangyi and Pandey, Vivek and Somarakis, Christoforos and Motee, Nader},
  booktitle={2022 American Control Conference (ACC)}, 
  title={Risk of Cascading Failures in Multi-agent Rendezvous with Communication Time Delay}, 
  year={2022},
  volume={},
  number={},
  pages={2172-2177}}

@incollection{sarykalin2008value,
  title={Value-at-risk vs. conditional value-at-risk in risk management and optimization},
  author={Sarykalin, Sergey and Serraino, Gaia and Uryasev, Stan},
  booktitle={State-of-the-art decision-making tools in the information-intensive age},
  pages={270--294},
  year={2008},
  publisher={Informs}
}

@article{somarakis2023risk,
  title={Risk of Phase Incoherence in Wide Area Control of Synchronous Power Networks with Time-Delayed and Corrupted Measurements},
  author={Somarakis, Christoforos and Liu, Guangyi and Motee, Nader},
  journal={IEEE Transactions on Automatic Control},
  year={2023},
  publisher={IEEE}
}

@INPROCEEDINGS{liu2023cascading,
  author={Liu, Guangyi and Pandey, Vivek and Somarakis, Christoforos and Motee, Nader},
  booktitle={2023 American Control Conference (ACC)}, 
  title={Cascading Waves of Fluctuation in Time-delay Multi-agent Rendezvous}, 
  year={2023},
  volume={},
  number={},
  pages={4110-4115}}

@book{horn2012matrix,
  title={Matrix analysis},
  author={Horn, Roger A and Johnson, Charles R},
  year={2012},
  publisher={Cambridge university press}
}

@article{liu2009trace,
  title={Trace inequalities for matrix products and trace bounds for the solution of the algebraic Riccati equations},
  author={Liu, Jianzhou and Zhang, Juan and Liu, Yu},
  journal={Journal of Inequalities and Applications},
  volume={2009},
  pages={1--17},
  year={2009},
  publisher={Springer}
}

@article{ren2007information,
  title={Information consensus in multivehicle cooperative control},
  author={Ren, W. and Beard, R. W. and Atkins, E. M.},
  journal={IEEE Control systems magazine},
  volume={27},
  number={2},
  pages={71--82},
  year={2007},
  publisher={IEEE}
}

@article{olfati2007consensus,
  title={Consensus and cooperation in networked multi-agent systems},
  author={Olfati-Saber, R. and Fax, J. A. and Murray, R. M.},
  journal={Proceedings of the IEEE},
  volume={95},
  number={1},
  pages={215--233},
  year={2007},
  publisher={IEEE}
}

@article{olfati2004consensus,
  title={Consensus problems in networks of agents with switching topology and time-delays},
  author={Olfati-Saber, R. and Murray, R. M.},
  journal={IEEE Transactions on automatic control},
  volume={49},
  number={9},
  pages={1520--1533},
  year={2004},
  publisher={IEEE}
}

@article{krueger1998perception,
  title={On the perception of social consensus},
  author={Krueger, J.},
  journal={Advances in experimental social psychology},
  volume={30},
  pages={163--240},
  year={1998},
  publisher={Elsevier}
}

@article{xie2011social,
  title={Social consensus through the influence of committed minorities},
  author={Xie, J. and Sreenivasan, Sameet and Korniss, Gyorgy and Zhang, Weituo and Lim, Chjan and Szymanski, Boleslaw K},
  journal={Physical Review E},
  volume={84},
  number={1},
  pages={011130},
  year={2011},
  publisher={APS}
}

@inproceedings{fagiolini2008consensus,
  title={Consensus-based distributed intrusion detection for multi-robot systems},
  author={Fagiolini, A. and Pellinacci, Marco and Valenti, Gianni and Dini, Gianluca and Bicchi, Antonio},
  booktitle={2008 IEEE International Conference on Robotics and Automation},
  pages={120--127},
  year={2008},
  organization={IEEE}
}

@article{rockafellar2002conditional,
    title = {{Conditional value-at-risk for general loss distributions}},
    year = {2002},
    journal = {Journal of Banking and Finance},
    author = {Rockafellar, R. Tyrrell and Uryasev, Stanislav},
    number = {7},
    pages = {1443--1471},
    volume = {26}
}

@article{Batson2014twice,
  title={Twice - Ramanujan Sparsifiers},
  author={J. Batson, D.A. Spielman and N. Srivastava},
  journal={SIAM Review},
  volume={56},
  number={2},
  pages={315--334},
  year={2014}
}

@article{diane1981schur,
  title={Schur Complements and Statistics},
  author={Diane Valerie Ouellette},
  journal={Linear Algebra and Its Applications},
  volume={36},
  number={9},
  pages={187--295},
  year={1981},
  publisher={Elsevier}
}

@inproceedings{liu2022emergence,
  title={Emergence of Cascading Risk and Role of Spatial Locations of Collisions in Time-Delayed Platoon of Vehicles},
  author={Liu, Guangyi and Somarakis, Christoforos and Motee, Nader},
  booktitle={2022 IEEE 61st Conference on Decision and Control (CDC)},
  pages={6460--6465},
  year={2022},
  organization={IEEE}
}

@inproceedings{saldana2018modquad,
  title={Modquad: The flying modular structure that self-assembles in midair},
  author={Saldana, David and Gabrich, Bruno and Li, Guanrui and Yim, Mark and Kumar, Vijay},
  booktitle={2018 IEEE International Conference on Robotics and Automation (ICRA)},
  pages={691--698},
  year={2018},
  organization={IEEE}
}

@book{greene2003econometric,
  title={Econometric analysis},
  author={Greene, William H},
  year={2003},
  publisher={Pearson Education India}
}

@article{Erdos1959,
  author    = {Paul Erd{\H{o}}s and Alfr{\'e}d R{\'e}nyi},
  title     = {On Random Graphs {I}},
  journal   = {Publicationes Mathematicae (Debrecen)},
  volume    = {6},
  pages     = {290--297},
  year      = {1959}
}

@article{liu2025risk,
  title={Risk of Cascading Collisions in Network of Vehicles with Delayed Communication},
  author={Liu, Guangyi and Somarakis, Christoforos and Motee, Nader},
  journal={IEEE Transactions on Automatic Control},
  year={2025},
  publisher={IEEE}
}

@inproceedings{liu2024data,
  title={Data-driven distributionally robust mitigation of risk of cascading failures},
  author={Liu, Guangyi and Amini, Arash and Pandey, Vivek and Motee, Nader},
  booktitle={2024 American Control Conference (ACC)},
  pages={3264--3269},
  year={2024},
  organization={IEEE}
}

@inproceedings{pandey2023quantification,
  title={Quantification of Distributionally Robust Risk of Cascade of Failures in Platoon of Vehicles},
  author={Pandey, Vivek and Liu, Guangyi and Amini, Arash and Motee, Nader},
  booktitle={2023 62nd IEEE Conference on Decision and Control (CDC)},
  pages={7401--7406},
  year={2023},
  organization={IEEE}
}

@article{pandey2025distributionally,
  title={Distributionally Robust Cascading Risk Quantification in Multi-Agent Rendezvous: Effects of Time Delay and Network Connectivity},
  author={Pandey, Vivek and Motee, Nader},
  journal={arXiv preprint arXiv:2507.23489},
  year={2025}
}

\appendix
\noindent \underline{\bf{Proof of Lemma \ref{lem:y_steady}:}}
The result is a immediate extension of the steady-state statistics of the observables in \cite{Somarakis2019g} by considering a centering matrix $M_n$.
\hfill$\square$

\vspace{1mm}

\noindent \underline{\bf{Proof of Theorem \ref{thm:gen_cas_ren_risk}:}}
Considering the fact that $|\bar{y}_i| \in U_{\delta^*}$, the conditional probability of $|\bar{y}_{j}| > z$ with $z \geq 0$ given that $|\bar{y}_i| \in U_{\delta^*}$ is:
\[
    \mathbb{P} \left\{|\bar{y}_{j}| > z \,\big|\, |\bar{y}_i| \in U_{\delta^*} \right\} = \frac{\mathbb{P} \{ |\bar{y}_{j}| > z \bigwedge |\bar{y}_i| \in U_{\delta^*}\}}{{\mathbb{P}\{|\bar{y}_i| \in U_{\delta^*}\}}},
\]
where
\begin{equation*}
\scalebox{0.9}{$
    \mathbb{P}\left(|\bar{y}_i| \in U_{\delta^*}\right)
    = \frac{2}{\sqrt{2\pi}\sigma_i}\int_{c\frac{\delta^*+1}{\delta^*+\alpha}}^\infty e^{-\frac{y^2}{2\sigma_i^2}}\,\mathrm{d}y
    = 1 - \textup{erf}\left(\frac{c(\delta^*+1)}{\sqrt{2}\sigma_i(\delta^*+\alpha)}\right).
$}
\end{equation*}

Using the bi-variate normal distribution probability density function, one has
\begin{equation} \label{eq:joint_prob_int}
    \scalebox{0.9}{$\begin{aligned}
        &\mathbb{P} \{ |\bar{y}_{j}| > z \wedge |\bar{y}_i| \in U_{\delta^*}\} = \\
        &\hspace{0.5cm} \frac{1}{2\pi\sigma_i \sigma_j \sqrt{1-\rho_{ij}^2}} \int_{|\bar{y}_j| > z} \int_{|\bar{y}_i| \in U_{\delta^*}} h(\bar{y}_{i},\bar{y}_{j}) \textrm{d} \bar{y}_{i} \textrm{d} \bar{y}_{j},
    \end{aligned}$}
\end{equation}
where 
\begin{equation*}
    \scalebox{0.8}{$
    \begin{aligned}
        h(\bar{y}_{i},\bar{y}_{j}) &= \exp\left(-\frac{1}{2(1-\rho^2_{ij})}\left[(\frac{\bar{y}_i}{\sigi})^2 + (\frac{\bar{y}_j}{\sigj})^2 - 2 \rho_{ij}\frac{\bar{y}_i \bar{y}_j}{\sigi \sigj} \right] \right)\\
        &=\exp\left(-\frac{1}{2(1-\rho^2_{ij})} \left[(1-\rho_{ij}^2) \left(\frac{\bar{y}_j}{\sigj} \right)^2 + \left(\frac{\bar{y}_i}{\sigi} - \rho_{ij}\frac{\bar{y}_j}{\sigj} \right)^2 \right] \right).
    \end{aligned}
    $}
\end{equation*}
Then, the integral inside \eqref{eq:joint_prob_int} can be simplified as 
\begin{equation}    \label{eq:joint_prob_inside}
    \begin{aligned}
        &\scalebox{0.85}{$
        \int_{|\bar{y}_j| > z} \exp\left(-\frac{1}{2} \left(\frac{\bar{y}_j}{\sigj} \right)^2 \right) \int_{|\bar{y}_i| \in U_{\delta^*}} \exp \left( -\frac{1}{2} \left(\frac{\bar{y}_i - \frac{\sigi\rho_{ij} \bar{y}_j}{\sigj}}{\sigi \sqrt{1-\rho_{ij}^2}}\right)^2\right) \textrm{d} \bar{y}_{i} \textrm{d} \bar{y}_{j}$}\\
        &\scalebox{0.9}{$
        = \int_{|\bar{y}_j| > z} \exp\left(-\frac{1}{2} \left(\frac{\bar{y}_j}{\sigj} \right)^2 \right) \left(1 - \frac{1}{2}\Psi_{-} (\bar{y}_{j}) + \frac{1}{2}\Psi_{+} (\bar{y}_{j})  \right) \textrm{d} \bar{y}_{j}
        $},
    \end{aligned}
\end{equation}
where $\Psi_{\pm}(\cdot)$ are defined explicitly in \eqref{eq:Psi_pm}. Then, the equation \eqref{eq:joint_prob_int} can be written as
\[
    \mathbb{P} \{ |\bar{y}_{j}| > z \wedge |\bar{y}_i| \in U_{\delta^*}\} = \Theta_-(z) + \Theta_+(z),
\]
with $\Theta_\pm (\cdot)$ defined in \eqref{eq:Theta_pm}. Since conditional distribution of $\bar{y}_j$ obtains a continuous density function, $\mathfrak{R}^{i,j}_{\varepsilon}$ can be computed by solving $\Theta_-(z) + \Theta_+(z) = \varepsilon$ for $z$, which does not obtain a convenient explicit form but can be evaluated numerically. Since the conditional distribution of $\bar y_j$ given $|\bar y_i|\in U_{\delta^*}$ admits a continuous density, the mapping $z\mapsto \Theta_-(z)+\Theta_+(z)$ is continuous and strictly decreasing on $[0,\infty)$, with limits $1$ and $0$ at $z=0$ and $z\to\infty$, respectively. Hence, for any $\varepsilon\in(0,1)$, a unique solution $\mathfrak{R}^{i,j}_\varepsilon$ exists. 
Then, by the definition of conditional $\AVAR_\varepsilon$,
\begin{align*}
    \mathfrak{A}^{i,j}_{\varepsilon}
    &= \mathbb{E}\!\left[|\bar{y}_{j}|\,\middle|\,|\bar{y}_{j}| > \mathfrak{R}^{i,j}_{\varepsilon},\ |\bar{y}_i| \in U_{\delta^*}\right] \\
    &=
    \frac{
    \mathbb{E}\!\left[
    |\bar{y}_{j}|\,\bm{1}_{\{|\bar{y}_{j}| > \mathfrak{R}^{i,j}_{\varepsilon},\,|\bar{y}_i| \in U_{\delta^*}\}}
    \right]
    }{
    \mathbb{P}\!\left(
    |\bar{y}_{j}| > \mathfrak{R}^{i,j}_{\varepsilon},\ |\bar{y}_i| \in U_{\delta^*}
    \right)
    }.
\end{align*}
Using the joint Gaussian density, the numerator is
\begin{equation*}
        \frac{1}{2\pi\sigma_i \sigma_j \sqrt{1-\rho_{ij}^2}}
        \int_{|\bar{y}_j| > \mathfrak{R}^{i,j}_{\varepsilon}}
        \int_{|\bar{y}_i| \in U_{\delta^*}}
        |\bar{y}_j|\, h(\bar{y}_{i},\bar{y}_{j}) \textrm{d} \bar{y}_{i} \textrm{d} \bar{y}_{j}.
\end{equation*}
Moreover, since
\[
    \mathbb{P}\!\left(|\bar{y}_{j}| > \mathfrak{R}^{i,j}_{\varepsilon} \,\middle|\, |\bar{y}_i| \in U_{\delta^*}\right)=\varepsilon,
\]
the denominator can be written as
\begin{equation*}
\begin{aligned}
    &\scalebox{0.85}{$
    \mathbb{P}\!\left(|\bar{y}_{j}| > \mathfrak{R}^{i,j}_{\varepsilon},\ |\bar{y}_i| \in U_{\delta^*} \right)
    = \mathbb{P}\!\left(|\bar{y}_{j}| > \mathfrak{R}^{i,j}_{\varepsilon} \,\middle|\, |\bar{y}_i| \in U_{\delta^*}\right) \mathbb{P}\!\left(|\bar{y}_i| \in U_{\delta^*}\right)$} \\
    & \hspace{3.4cm} \scalebox{0.85}{$
    = \varepsilon \left[1 - \textup{erf}\left(\frac{c(\delta^*+1)}{\sqrt{2}\sigma_i(\delta^*+\alpha)}\right)\right].
    $}
\end{aligned}
\end{equation*}
Substituting the numerator and denominator yields \eqref{eq:range_bounded_avar}. The expression of $\cass$ then follows by using \eqref{eq:risk_def_cond}. \hfill$\square$

\vspace{0.1mm}

\noindent \underline{\bf{Proof of Lemma \ref{lem:multi_conditional_prob}:}}
The result follows directly after Lemma \ref{lem:y_steady} and the conditional distribution of a multi-variate normal random variable as in \cite{tong2012multivariate}. \hfill$\square$

\vspace{0.1mm}
\noindent \underline{\textbf{Proof of Theorem \ref{thm:mul_cas_ren_risk}:}}
Using the result obtained from \eqref{eq:lem_mul_cond} and the cumulative distribution function of the folded normal distribution, the evaluation of $\casv$ is given by
\begin{align*}
    \casv = \inf \left\{z \,\Big|\, 1 - \frac{1}{2} \left(\erf(\frac{z-\tm}{\sqrt{2\ts^2}}) + \erf(\frac{z+\tm}{\sqrt{2\ts^2}})\right) < \varepsilon \right\},
\end{align*}
which, given the continuous nature of the density function, can be obtained by solving $1 - \frac{1}{2} (\erf(\frac{z-\tm}{\sqrt{2\ts^2}}) + \erf(\frac{z+\tm}{\sqrt{2\ts^2}})) = \varepsilon$ for $z$.
Then, following \eqref{eq:mul_avar}, the value of $\mathfrak{A}^{\mathcal{I}_m,j}_\varepsilon$ is given by
\begin{align*}
    \mathbb{E}[|y| \, \big| \, |y| > \mathfrak{R}^{\mathcal{I}_m,j}_\varepsilon]
    &= \frac{\mathbb{E}\!\left[|y| \cdot \bm{1}_{\{|y| > \mathfrak{R}^{\mathcal{I}_m,j}_\varepsilon\}}\right]}
    {\mathbb{P}\!\left(|y| > \mathfrak{R}^{\mathcal{I}_m,j}_\varepsilon\right)} \\
    &\hspace{-3cm} =\frac{1}{\sqrt{2\pi} \varepsilon\tilde{\sigma}} \int_{\mathfrak{R}^{\mathcal{I}_m,j}_\varepsilon}^{\infty} y \bigg( e^{-\frac{(y -\Tilde{\mu})^2}{2\Tilde{\sigma}^2}} + e^{-\frac{(y + \Tilde{\mu})^2}{2\Tilde{\sigma}^2}}\bigg) \text{d} y ,
\end{align*}
where $\bm{1}_{\{ \cdot \}}$ denotes the indicator function. Using the result from \cite{greene2003econometric} (Theorem 22.2), one has
\begin{equation*}
    \scalebox{0.8}{$
    \begin{aligned}
        \mathfrak{A}^{\mathcal{I}_m,j}_\varepsilon &= \frac{1}{\varepsilon} \bigg[ \frac{\tilde{\mu}}{2}\left(\textup{erf} \left(\frac{\casv + \tilde{\mu}}{\sqrt{2}\tilde{\sigma}} \right) -\textup{erf}\left(\frac{\casv - \tilde{\mu}}{\sqrt{2}\tilde{\sigma}} \right) \right)\\
        &\hspace{3cm} + \frac{\tilde{\sigma}}{\sqrt{2\pi}} \left(e^{-\left(\frac{\casv + \tilde{\mu}}{\sqrt{2}\tilde{\sigma}} \right)^2 } + e^{-\left(\frac{\casv - \tilde{\mu}}{\sqrt{2}\tilde{\sigma}} \right)^2 }\right) \bigg].
    \end{aligned}
    $}
\end{equation*}
Then, one can compare $\casa$ with $c$ and $\frac{c}{\alpha}$ to conclude the conditions for cases when $\casm = 0$ and $\casm = \infty$. When $\casa \in (\frac{c}{\alpha}, c)$, one has $\casm = \frac{\alpha\casa-c}{c-\casa}$.
\hfill$\square$

\vspace{0.1mm}
\noindent \underline{\textbf{Proof of Theorem \ref{thm:conditional_prob_update}}:}
To characterize the effect of the failures of $m + 1$ agents, let us consider the block covariance matrix
$$
\tilde{\Sigma}_{22}' = \begin{bmatrix}
\tilde{\Sigma}_{22} & \tilde{\Sigma}_{21}(k)\\
\tilde{\Sigma}_{12}(k) & \tilde{\Sigma}_{kk}
\end{bmatrix},
$$
where $\tilde{\Sigma}_{kk} = \sigma _k^2, \tilde{\Sigma}_{21}(k) = \tilde{\Sigma}_{12}^\top(k) =  \begin{bmatrix}
\sigma_{ki_1} & \dots  \sigma_{ki_m}
\end{bmatrix}$, and 
$\tilde{\Sigma}_{22}$ is obtained from \eqref{eq:block_cov}. Since $\tilde{\Sigma}_{22}$ is invertible, we have 
\begin{equation*}
    \scalebox{0.9}{$
        \tilde{\Sigma}_{22}'^{-1} = \begin{bmatrix}
        \tilde{\Sigma}_{22}^{-1} (I_m+ \tilde{\Sigma}_{21}(k) S^{-1} \tilde{\Sigma}_{12}(k) \tilde{\Sigma}_{22}^{-1})  & - \tilde{\Sigma}_{22}^{-1}\tilde{\Sigma}_{21}(k)S^{-1}\\
        S^{-1} \tilde{\Sigma}_{12}(k) \tilde{\Sigma}_{22}^{-1} & S^{-1}
        \end{bmatrix},
    $}
\end{equation*}
where 
$
    S = \tilde{\Sigma}_{22}'/\tilde{\Sigma}_{22} =  \tilde{\Sigma}_{kk} -  \tilde{\Sigma}_{12}(k)\tilde{\Sigma}_{22}^{-1}\tilde{\Sigma}_{21}(k) = \tilde{\sigma}_k^2 
$
is the Schur complement \cite{diane1981schur} of block $\tilde{\Sigma}_{22}$ of the matrix $\tilde{\Sigma}_{22}'$. Let us consider the vector of failed observables of $(m+1)$ agents as
$
[
\bar{\bm{y}}_f ~ \bar{{y}}_{f_k}
]^\top,$ where $\bm{\bar{y}}_{f} = [
\bar{y}_{f_1},...,\bar{y}_{f_m}
]^\top$ is the vector of failed observables of $m$ agents and $\bar{{y}}_{f_k}$ is the failed observable of agent k, i.e., $(m+1)^{th}$ agent.
Consider the following vectors, 
$\tilde{\Sigma}_{12}' = [\tilde{\Sigma}_{12} ~ \tilde{\Sigma}_{12}(k)] = \tilde{\Sigma}_{12}'^{T}$ and the conditional cross-covariance of agents $j$ and $k$ after $m$ agents have failed $\tilde{\sigma}_{jk} = \sigma_{jk} - \tilde{\Sigma}_{12}(j) \tilde{\Sigma}_{22}^{-1} \tilde{\Sigma}_{21}(k)$,   
the result then follows directly by applying Lemma \ref{lem:multi_conditional_prob}. \hfill$\square$

\vspace{0.1mm}
\noindent \underline{\textbf{Proof of Lemma \ref{lem:sig_complete}:}}
Using the result of Lemma \ref{lem:y_steady} and considering the eigenvalues of the complete graph $\lambda_i = n$ for any $i \in \{ 2, \dots, n\}$. For the case of $i \neq j$, 
\begin{equation*}
    \scalebox{0.8}{$
    \begin{aligned}
    \sigma_{ij} &= \frac{1}{2} b^2 \frac{\cos (n \tau)}{n (1 - \sin (n \tau))} \sum_{k=2}^{n}  (\bm{m}_i^\top \bm{q}_k)(\bm{m}_j^\top \bm{q}_k)\\ 
    &= \frac{1}{2} b^2 \frac{\cos (n \tau)}{n (1 - \sin (n \tau))} \left((QQ^\top)_{ij} - \bm{q}_1^\top\bm{q}_1\right) = -\frac{b^2}{2n^2} \frac{\cos (n \tau)}{1 - \sin (n \tau)}.
    \end{aligned}
    $}
\end{equation*}
Then, the result for $i = j$ follows similarly by considering the fact that $Q$ is an orthogonal matrix. 
\hfill$\square$

\vspace{0.1mm}
\noindent \underline{\textbf{Proof of Lemma \ref{lem:complete_mu_sig_cond}:}} The structure of $\tilde{\Sigma}_{22}$ in a complete graph is a special case of the Toeplitz matrix \cite{gray2006toeplitz} where the off-diagonal elements are identical. In addition, $\tilde{\Sigma}_{22}$ can be written as sum of a diagonal matrix and a rank one matrix, such that
\begin{equation*}
    \scalebox{1}{$
        \tilde{\Sigma}_{22} =  \frac{(n-1)b^2}{2n^2}\frac{1-\sin(n\tau)}{ \cos(n\tau)}\left((1-\rho)\mathbf{I}_{m} + \rho \mathbf{1}_{m} \mathbf{1}_{m}^\top\right),
    $}
\end{equation*}
where $\rho = \frac{1}{1-n}$. Then, one can apply the Sherman-Morrison Formula \cite{Batson2014twice} to obtain
\begin{equation*}
    \scalebox{0.9}{$
    \begin{aligned}
        \tilde{\Sigma}_{22}^{-1} &= \frac{2n^2}{b^2 (n-1)}\frac{1-\sin(n\tau)}{\cos(n\tau)} \frac{1}{1-\rho}\left(\mathbf{I}_{m} - \frac{\rho}{1 + ({m} - 1) \rho}\mathbf{1}_{m} \mathbf{1}_{m}^\top\right)\\
        & = \frac{2n}{b^2}\frac{1-\sin(n\tau)}{\cos(n\tau)}\left(\mathbf{I}_{m} + \frac{\mathbf{1}_{m} \mathbf{1}_{m}^\top}{n-m} \right),
    \end{aligned}
    $}
\end{equation*}
which is well-defined since $m < n$. Then, the result follows immediately by applying Lemma \ref{lem:multi_conditional_prob}. \hfill$\square$

\vspace{0.1mm}
\noindent \underline{\textbf{Proof of Lemma \ref{lem:sig_star}}:} The proof is a direct result of Lemma \ref{lem:y_steady} by considering the eigenvalues of the star graph topology $\lambda_i = 1$ for any $i \in \{ 2, \dots, n - 1\}$ and $\lambda_n = n$. With some basic algebraic calculations, the covariance term can be written as
\begin{equation*}
    \begin{aligned} 
    \sigma_{ij} &= \frac{1}{2} b^2 \bigg[g(\tau) \left((QQ^\top)_{ij} - (\bm{q}_1\bm{q}_1^\top)_{ij} - (\bm{q}_n\bm{q}_n^\top)_{ij} \right) \\
    &\hspace{2cm} + \frac{1}{n}g(n\tau)(\bm{q}_n\bm{q}_n^\top)_{ij}\bigg], 
    \end{aligned}
\end{equation*}
where $g(x)$ is as defined in $\eqref{eq:g}$
For the case of $i \neq j$ and all $i, j \in \{ 1, \dots, n - 1\}$ 
\begin{equation*}
    \scalebox{1}{$
        \sigma_{ij} = \frac{1}{2} b^2 \left(g(\tau)\left(0 - \frac{1}{n} - \frac{1}{n(n-1)}\right) + \frac{1}{n} g(n\tau)\frac{1}{n(n-1)}\right).
    $}
\end{equation*}
The result follows by simplification. The result for other cases follows similarly by considering the fact that $Q$ is an orthogonal matrix and using appropriate elements of the matrices $(\bm{q}_1\bm{q}_1^\top)$ and $(\bm{q}_n\bm{q}_n^\top)$.
\hfill$\square$

\vspace{0.1mm}
\noindent \underline{\bf{Proof of Lemma \ref{lem:star_mu_sig}:}}
The proof is similar to the proof of Lemma \ref{lem:complete_mu_sig_cond}. The $\tilde{\Sigma}_{22}^{-1}$ is calculated using the inverse formula involving Schur complement.  
\hfill$\square$

\vspace{0.1mm}
\noindent\underline{\bf Proof of Lemma \ref{lem:f_lower_bound}:}
\cite{Somarakis2019g} Let $f(x)=\frac{1}{2x}\frac{\cos x}{1-\sin x}$ on $(0,\frac{\pi}{2})$. Since $f$ is smooth on this open interval and $\lim_{x\downarrow 0}f(x)=+\infty$ while $\lim_{x\uparrow \frac{\pi}{2}}f(x)=+\infty$, it attains a finite minimum at some $x^\star\in(0,\frac{\pi}{2})$. Differentiating and setting $f'(x)=0$ yields a unique critical point $x^\star$ in $(0,\frac{\pi}{2})$ (solvable numerically). Evaluating $f(x^\star)$ gives the stated value $\underline f\approx 1.5319$. \hfill$\square$

\vspace{0.1mm}
\noindent \underline{\textbf{Proof of Theorem \ref{thm:sigma_bound}:}}
Let us consider 
\begin{equation*}
        f(\lambda_i \tau) = \frac{1}{2\lambda_i\tau} \frac{\cos(\lambda_i \, \tau)}{1 - \sin(\lambda_i \, \tau)},
\end{equation*}
and rewrite \eqref{eq:sigma_y} as 
\begin{equation*}
    \scalebox{0.95}{$
        \begin{aligned}
        \sigma_{ij} &= b^2 \tau \Tr \big(\text{diag}\{\bm{m}_i^\top\bm{q}_1\bm{m}_j^\top\bm{q}_1 \, \underline{f}, \bm{m}_i^\top\bm{q}_2\bm{m}_j^\top\bm{q}_2 f(\lambda_2 \tau),\\
        &\hspace{3cm} ...,\bm{m}_i^\top\bm{q}_n\bm{m}_j^\top\bm{q}_n f(\lambda_n \tau)\} \big)\\
        & = b^2 \tau \Tr \big(\text{diag}\left\{\underline{f},f(\lambda_2  \tau),...,f(\lambda_n \tau)\right\} \times \\
        & \hspace{3cm} \text{diag}\left\{ \bm{m}_i^\top\bm{q}_1\bm{m}_j^\top\bm{q}_1,...,\bm{m}_i^\top\bm{q}_n\bm{m}_j^\top\bm{q}_n \right\}\big)\\
        & = b^2 \tau \Tr ( F E_{ij} ) = b^2 \tau \Tr ( E_{ij}F ),
    \end{aligned}
    $}
\end{equation*}
where $F = \text{diag}\left\{\underline{f},f(\lambda_2 \tau),...,f(\lambda_n \tau)\right\}$, and $E_{ij} = (\bm{m}_i^\top Q)^\top \bm{m}_j^\top Q$. Since $\bm{m}_i^\top\bm{q}_1\bm{m}_j^\top\bm{q}_1 = 0$ always holds, we can set $(F)_{11} = \underline{f}$, the lower bound of $f$, without loss of generality. Considering the fact that $F$ is a normal matrix \cite{horn2012matrix} in $\R^{n \times n}$, we can use the result from \cite{liu2009trace} (Theorem 2.10) to show
\vspace{-2mm}
\[
    \sum_{i = 1}^{n} \Re(\gamma_{n-i+1}(F))\mu_i(\bar{E}_{ij}) \leq \Tr(E_{ij}F),
\] 
and
\[
    \Tr(E_{ij}F) \leq \sum_{i =  1}^{n} \Re(\gamma_{n-i+1}(F))\mu_{n-i+1}(\bar{E}_{ij}),
\]
where $\gamma_{i}(\cdot)$ and $\mu_{i}(\cdot)$ denotes the $i$'th eigenvalue of the matrix $F$ and $\bar{E}_{ij}$ in the non-decreasing order, and $\bar{E}_{ij} = (E_{ij}+E_{ij}^\top)/2$. Let us denote by the eigenvalues of $F$ as $\lambda_i(F)$, the smallest and the largest eigenvalue as $\gamma_{min}(F)$ and $\gamma_{max}(F)$. Use the convexity of the $f(\cdot)$ from \cite{ghaedsharaf2016interplay}. Then, the above inequality can be written as
\begin{equation*}
    \sum_{i = 1}^{n} \lambda_i(F) \mu_i(\bar{E}_{ij}) \leq  \Tr(E_{ij}F) \leq \sum_{i = 1}^{n} \lambda_i(F) \mu_{n-i+1}(\bar{E}_{ij}).
\end{equation*}
Considering the fact that,
\[
    \bar{E}_{ij} = \frac{E_{ij}+E_{ij}^\top}{2} = \frac{Q^\top(\bm{m}_i\bm{m}_j^\top+\bm{m}_j\bm{m}_i^\top)Q}{2},
\]
and $\mu_i(\bar{E}_{ij}) = \mu_i(QQ^\top \frac{\bm{m}_i\bm{m}_j^\top+\bm{m}_j\bm{m}_i^\top}{2}) = \mu_i(\frac{\bm{m}_i\bm{m}_j^\top+\bm{m}_j\bm{m}_i^\top}{2})$. By observing the structure of $\tilde{E}_{ij} = \frac{\bm{m}_i\bm{m}_j^\top+\bm{m}_j\bm{m}_i^\top}{2}$, the eigenvalues of $\tilde{E}_{ij}$ can be simplified as follows.

\vspace{1mm}
\noindent{\underline{Case 1}:} When $|i-j| = 0$, $\tilde{E}_{ii} = \bm{m}_i\bm{m}_i^\top$ is a positive semi definite rank one matrix, which has only one non-zero eigenvalue given by 
$\mu_1 = \bm{m}_i^\top\bm{m}_i = 1 - \frac{1}{n},$ where 
$(\bm{m}_{i})_j = -1/n $ for $j\neq i$ and $(\bm{m}_{i})_i = 1 -1/n$. Then, we have $\mu_1 = (n-1)/n$ and $\mu_2 = ... = \mu_n = 0$.

\vspace{1mm}
\noindent{\underline{Case 2}:} When $|i-j| \geq 1$, $\tilde{E}_{ij} = \frac{\bm{m}_i\bm{m}_j^\top+\bm{m}_j\bm{m}_i^\top}{2}$ is a rank two matrix, all but two of its eigenvalues are zero. The eigenspace of dimension two is spanned by the the columns of each rank one term in $\tilde{E}_{ij}.$ For constants $\alpha, \beta \in \R,$ let the eigenvectors be $\alpha \bm{m}_i + \beta \bm{m}_j$, we have
\begin{align*}
    \tilde{E}_{ij} \, v = \frac{\bm{m}_i\bm{m}_j^\top+\bm{m}_j\bm{m}_i^\top}{2} \big(\alpha \bm{m}_i + \beta \bm{m}_j \big).
\end{align*}

To find eigenvalue $\mu$, we have
\begin{align*}
    (\tilde{E}_{ij} - \mu I )v = \left(\frac{\bm{m}_i\bm{m}_j^\top+\bm{m}_j\bm{m}_i^\top}{2} - \mu I \right)(\alpha \bm{m}_i + \beta \bm{m}_j) = 0. 
\end{align*}
Rearranging the R.H.S leads to 
\begin{equation}    \label{eq:E_case_2}
    \begin{aligned}
        \bm{m}_i\left(\frac{\alpha \bm{m}_j^\top\bm{m}_i + \beta \|\bm{m}_j\|^2}{2} - \alpha \, \mu I\right) &+ \\    & \hspace{-4cm}\bm{m}_j\left(\frac{\beta \bm{m}_j^\top\bm{m}_i + \alpha \|\bm{m}_i\|^2}{2} - \beta \, \mu I\right) = 0.
    \end{aligned}
\end{equation}
Since $i \neq j$, $\bm{m}_i$ and $\bm{m}_j$ are linearly independent, which  implies
\begin{align*}
    \frac{\alpha \bm{m}_j^\top\bm{m}_i + \beta \|\bm{m}_j\|^2}{2} - \alpha \mu  &=0 ,\\     
    \frac{\beta \bm{m}_j^\top\bm{m}_i + \alpha \|\bm{m}_i\|^2}{2} - \beta \mu   &= 0,
\end{align*}
which can be written as 
\begin{align*}
    \begin{bmatrix}
       \frac{\bm{m}_j^\top\bm{m}_i}{2}  - \mu  & \frac{\|\bm{m}_j\|^2}{2} \\
       \frac{\|\bm{m}_i\|^2}{2} & \frac{\bm{m}_j^\top\bm{m}_i}{2}  - \mu
    \end{bmatrix} \begin{bmatrix}
        \alpha\\
        \beta
    \end{bmatrix}= \begin{bmatrix}
        0 \\
        0
    \end{bmatrix}.
\end{align*}
Since the vector $[\alpha;\beta]$ lies in the kernel of the coefficient matrix and its determinant must be zero, such that
\begin{align*}
    \mu_1 &= \frac{\bm{m}_j^\top\bm{m}_i + \|\bm{m}_i\| \|\bm{m}_j\|}{2},
    \mu_n &= \frac{\bm{m}_j^\top\bm{m}_i - \|\bm{m}_i\| \|\bm{m}_j\|}{2} 
\end{align*}
Substituting the values of $\|\bm{m}_i\|, \|\bm{m}_j\|, \bm{m}_j^\top\bm{m}_i$ leads to $\mu_1 = (n-2)/2n$, $\mu_2 = ... = \mu_{n-1} = 0$, $\mu_n = -1/2$.

Then, by combining the eigenvalues of $\tilde{E}_{ij}$, $\gamma_{min}(F) = \underline{f}$ as in \cite{Somarakis2019g}, and $\gamma_{max}(F) = \bar{f}$ for the convex and compact subset $\bar{S}$, one can conclude the result.
\hfill$\square$

\noindent\underline{\bf Proof of Theorem \ref{thm:risk_bound}:}
Define the variance band
\[
\sigma_{\min}:=\sqrt{\tfrac{n-1}{n}\,b^2\tau\,\underline f},\qquad
\sigma_{\max}:=\sqrt{\tfrac{n-1}{n}\,b^2\tau\,\bar f},
\]
where $\underline f:=\inf_{x \in \bar{S}} \frac{\cos x}{2x(1-\sin x)}$ and
$\bar f:=\sup_{x\in\bar S} \frac{\cos x}{2x(1-\sin x)}$ (finite under Assumption~\ref{asp:stable} and the domain choice $\bar S$). By Theorem~\ref{thm:sigma_bound}, for any connected graph satisfying the delay stability we have
\[
\sigma_i,\sigma_j\in[\sigma_{\min},\sigma_{\max}].
\]
Moreover, positive semidefiniteness of $\Sigma$ implies the 2$\times$2 principal minor condition
\(
\begin{bmatrix}\sigma_i^2 & \sigma_{ij}\\ \sigma_{ij} & \sigma_j^2\end{bmatrix}\succeq 0
\),
hence
\[
|\sigma_{ij}|\le \sigma_i\sigma_j.
\]
We collect these into the feasible set
\[
\mathbb W_1:=\Big\{(\sigma_i,\sigma_j,\sigma_{ij}) : \sigma_{i,j}\in[\sigma_{\min},\sigma_{\max}],\ |\sigma_{ij}|\le \sigma_i\sigma_j\Big\},
\]
and lower bound the folded tail $\AVAR_\varepsilon$ by the unfolded surrogate
\[
\underline{\mathfrak A}_\varepsilon(\sigma_i,\sigma_j,\sigma_{ij})
=\frac{\sigma_{ij}}{\sigma_i^2}\,y_f
+\kappa_\varepsilon\sqrt{\sigma_j^2-\frac{\sigma_{ij}^2}{\sigma_i^2}},\]
where $\kappa_\varepsilon:=\big(\sqrt{2\pi}\,\varepsilon\,e^{\iota_\varepsilon^2}\big)^{-1} \text{ and } \iota_\varepsilon:=\erf^{-1}(2\varepsilon-1)$, so that $\underline{\mathfrak A}_\varepsilon\le \mathfrak A^{ij}_\varepsilon$ for fixed $(\sigma_i,\sigma_j,\sigma_{ij})$.
For fixed $\sigma_i,\sigma_j$, the function $\underline{\mathfrak A}_\varepsilon$ is strictly concave in $\sigma_{ij}$ on the interval $[-\sigma_i\sigma_j,\sigma_i\sigma_j]$ because
\[
\frac{\partial^2\underline{\mathfrak A}_\varepsilon}{\partial\sigma_{ij}^2}
=-\,\frac{\kappa_\varepsilon\,\sigma_j^2}{\sigma_i^2\big(\sigma_j^2-\sigma_{ij}^2/\sigma_i^2\big)^{3/2}}<0.
\]
Therefore, its minimum over $\sigma_{ij}$ is attained at an endpoint of that interval.

\emph{Case 1: $\sigma_{ij}>0$.}
Here the feasible interval is $[0,\sigma_i\sigma_j]$.
Evaluating at the endpoints gives
\[
\underline{\mathfrak A}_\varepsilon(0)=\kappa_\varepsilon\,\sigma_j,\qquad
\underline{\mathfrak A}_\varepsilon(\sigma_i\sigma_j)=\frac{\sigma_j}{\sigma_i}\,y_f.
\]
Hence
\[
\min_{\sigma_{ij}\in[0,\sigma_i\sigma_j]}\underline{\mathfrak A}_\varepsilon
=\min\!\Big\{\kappa_\varepsilon\,\sigma_j,\ \frac{\sigma_j}{\sigma_i}\,y_f\Big\}.
\]
Minimizing further over $\sigma_i,\sigma_j\in[\sigma_{\min},\sigma_{\max}]$ yields
\begin{align*}
    \inf_{\mathbb W_1}\underline{\mathfrak A}_\varepsilon &=\min\!\Big\{\kappa_\varepsilon\,\sigma_{\min},\ \frac{\sigma_{\min}}{\sigma_{\max}}\,y_f\Big\}\\
    &=\min\!\Big\{\kappa_\varepsilon\,\sigma_{\min},\ \sqrt{\tfrac{\underline f}{\bar f}}\,y_f\Big\}
=: \mathfrak A_+.
\end{align*}

Since $\mathfrak A^{ij}_\varepsilon\ge \underline{\mathfrak A}_\varepsilon$, we obtain the stated $\mathfrak A_+$ and the branch mapping to $\mathcal A_+$ via \eqref{eq:casa}.

\emph{Case 2: $\sigma_{ij}<0$.}
Here the feasible interval is $[-\sigma_i\sigma_j,0]$. The endpoint values are
\[
\underline{\mathfrak A}_\varepsilon(0)=\kappa_\varepsilon\,\sigma_j\ge 0,\qquad
\underline{\mathfrak A}_\varepsilon(-\sigma_i\sigma_j)=-\frac{\sigma_j}{\sigma_i}\,y_f\le 0,
\]
so the minimum is nonpositive. Consequently, any lower bound on $\AVAR_\varepsilon$ is $\le 0$, and the level-set mapping \eqref{eq:casa} gives $\mathcal A^{ij}_\varepsilon\ge 0$, i.e., $\mathcal A_-=0$.

\emph{Case 3: $\sigma_{ij}=0$.}
When the correlation vanishes, the folded tail reduces to the one-dimensional case, giving
\[
\mathfrak A^{ij}_\varepsilon=\kappa_{\varepsilon/2}\,\sigma_j\ \ \Rightarrow\ \ 
\mathfrak A_0=\kappa_{\varepsilon/2}\,\sigma_{\min},
\]
and applying \eqref{eq:casa} yields the stated branch for $\mathcal A_0$.

\medskip
Putting the three cases together and then applying the level-set mapping \eqref{eq:casa} completes the proof. \hfill$\square$

\vspace{0.1mm}
\noindent \underline{\textbf{Proof of Corollary \ref{cor:best_risk_complete}}:} Observing the result from Lemma \ref{lem:complete_mu_sig_cond}, one can notice that the conditional mean $\tilde{\mu}$ is independent to the graph structure, then the best achievable risk of cascading collision can be obtained at $\sigj = \sqrt{\frac{(n-1) b^2\tau}{n} \underline{f}}$. Then, we conclude the result by inserting the conditional statistics into Theorem \ref{thm:mul_cas_ren_risk}.
\hfill$\square$


\begin{IEEEbiography}[{\includegraphics[width=1in,height=1.25in,clip,keepaspectratio]{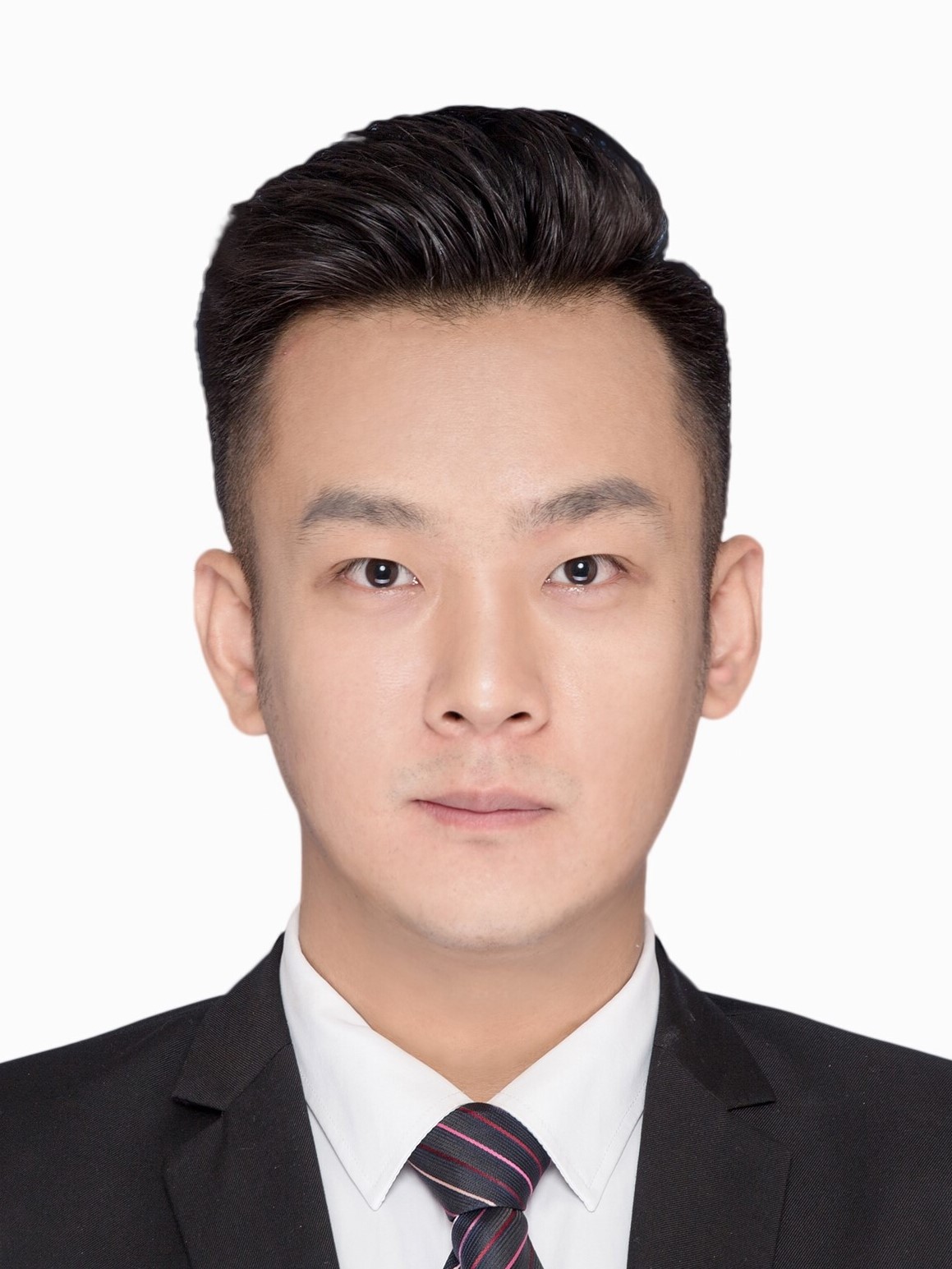}}]{Guangyi Liu}
    Guangyi Liu received the B.E. degree in aircraft design and engineering from the Beijing Institute of Technology in 2016, and the M.S. and Ph.D. degrees in mechanical engineering from Lehigh University in 2018 and 2024, respectively. He is currently a Postdoctoral Research Scientist at Amazon Robotics. His research interests include risk-aware decision-making, perception and networked control systems, and reinforcement learning for large-scale autonomous systems.
\end{IEEEbiography}

\begin{IEEEbiography}
[{\includegraphics[width=1in,height=1.25in,clip,keepaspectratio]{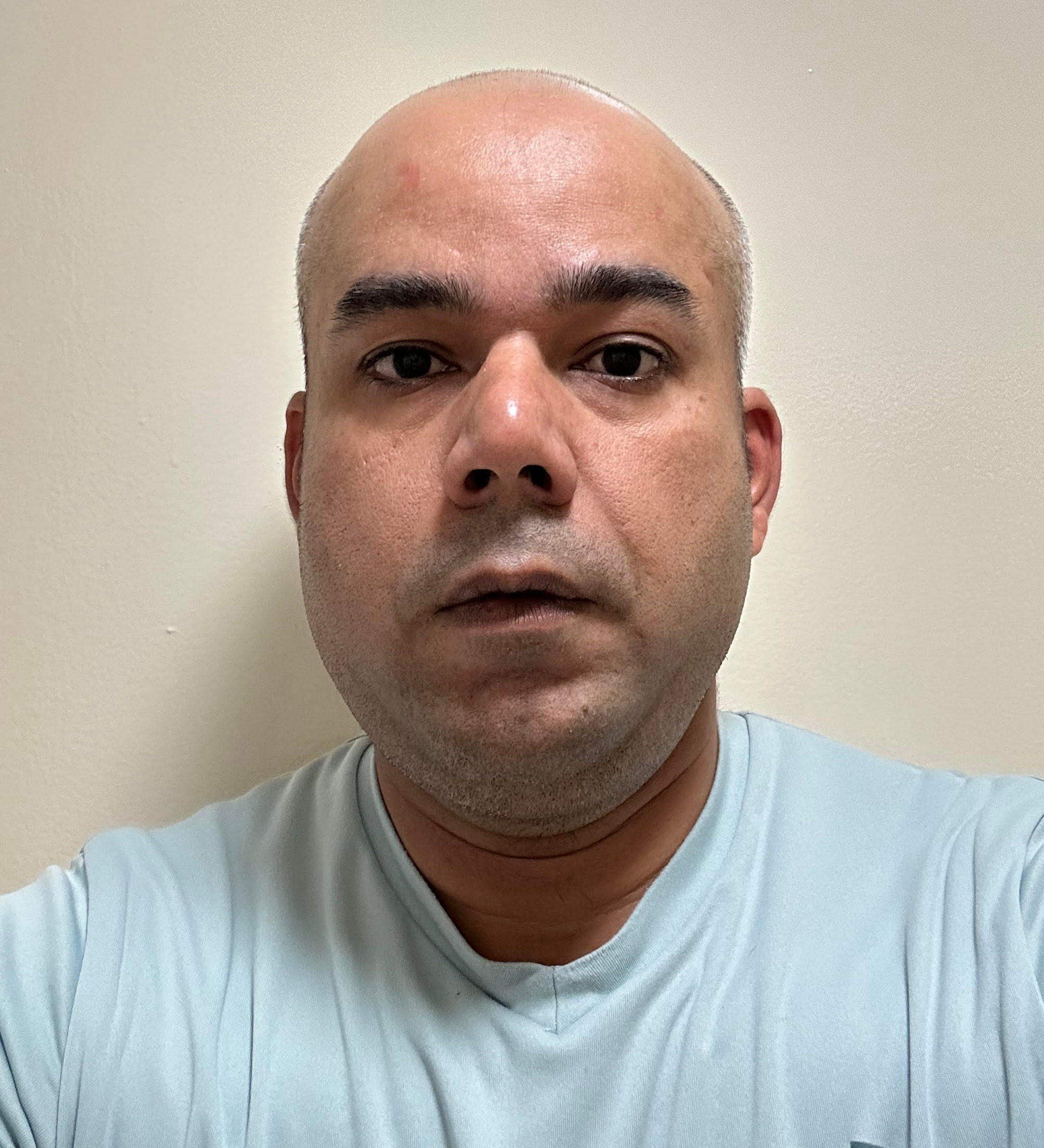}}]
{Vivek Pandey}
Vivek Pandey received his B.Tech and M.Tech degree in Chemical Engineering from Indian Institute of Technology, Mumbai, India in 2014. He is currently pursuing a Ph.D. degree in the Department of Mechanical Engineering and Mechanics at Lehigh University. His research interests include networked control systems.
\end{IEEEbiography}

\begin{IEEEbiography}[{\includegraphics[width=1in,height=1.25in,clip,keepaspectratio]{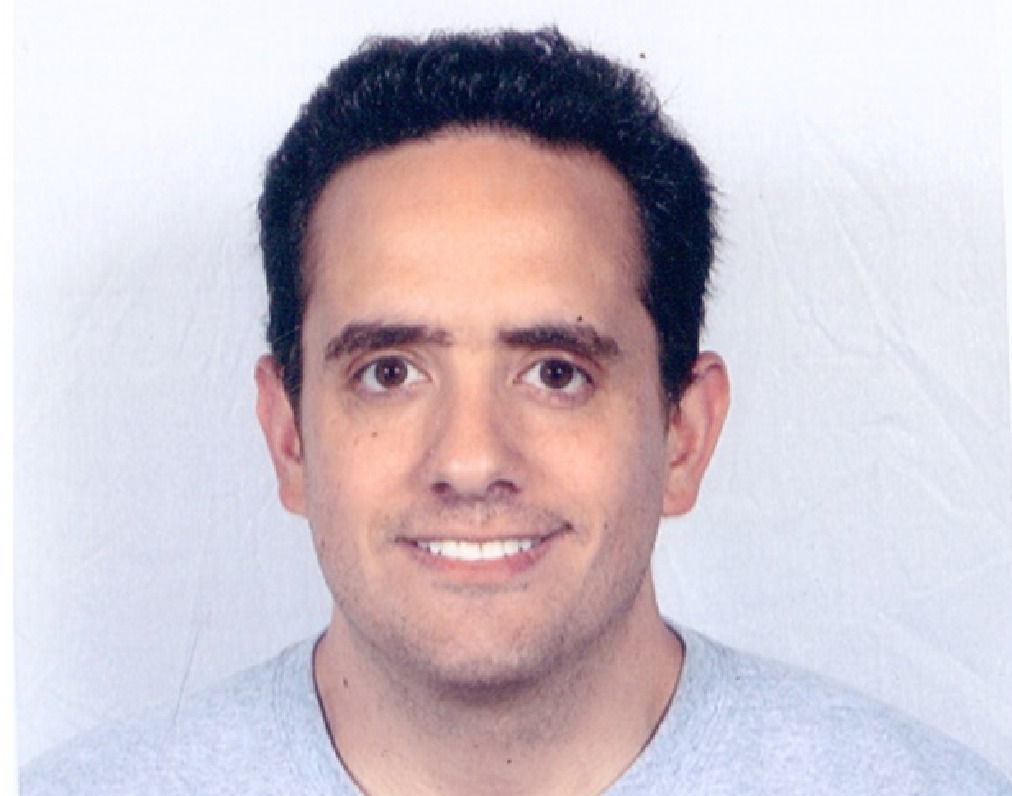}}]{Christoforos Somarakis}
Christoforos Somarakis received the B.S. degree
in Electrical Engineering from the National Technical
University of Athens, Athens, Greece, in 2007
and the M.S. and Ph.D. degrees in applied
mathematics from the University of Maryland at
College Park, in 2012 and 2015, respectively. He
was a Post-Doctoral scholar and a Research Scientist
with the Department of Mechanical Engineering and
Mechanics at Lehigh University from 2016 to 2019. From 2019 until 2022 he was a Member of Research Staff with the Intelligent Systems Lab at Palo Alto Research Center - Xerox. He is currently a Senior Scientits of Mathematical Modelling with the Applied Mathematics Group at Merck \& Co. 
\end{IEEEbiography}

\begin{IEEEbiography}[{\includegraphics[width=1in,height=1.25in,clip,keepaspectratio]{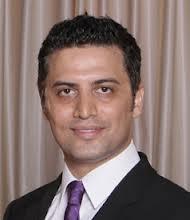}}]{Nader Motee}
Nader Motee (Senior Member, IEEE) received 
the B.Sc. degree in electrical engineering from 
the Sharif University of Technology, Tehran, 
Iran, in 2000, and the M.Sc. and Ph.D. degrees 
in electrical and systems engineering from the 
University of Pennsylvania, Philadelphia, PA, 
USA, in 2006 and 2007, respectively. 
From 2008 to 2011, he was a Postdoctoral 
Scholar with the Control and Dynamical Systems Department, California Institute of Technology, Pasadena, CA, USA. He is currently a 
Professor with the Department of Mechanical Engineering and Mechanics, Lehigh University, Bethlehem, PA, USA. His research interests
include distributed control systems and real-time robot perception. 
Dr. Motee was the recipient of several awards including the 2019 Best 
SIAM Journal of Control and Optimization Paper Prize, the 2008 AACC 
Hugo Schuck Best Paper Award, the 2007 ACC Best Student Paper 
Award, the 2008 Joseph and Rosaline Wolf Best Thesis Award, the 
2013 Air Force Office of Scientific Research Young Investigator Program 
Award, 2015 NSF Faculty Early Career Development Award, and a 2016 
Office of Naval Research Young Investigator Program Award.
 \end{IEEEbiography}

\end{document}